\documentclass[12pt,smallextended,numbook,runningheads]{svjour3}     
\usepackage[sort,nocompress]{cite}
\smartqed  
\usepackage{graphicx}
\usepackage{amsmath}
\usepackage{epstopdf} 

\usepackage{mathptmx}      
\usepackage{amsmath,amssymb,blkarray,color,enumerate,listings,mathtools,tabularx,url,xspace,array,amsfonts,xcolor}

\usepackage[justification=justified,singlelinecheck=true]{caption}
\usepackage{subcaption}

\usepackage{multicol,multirow}        
\usepackage[bottom]{footmisc}        
\usepackage[sort,nocompress]{cite}
\usepackage{hyperref}
  \hypersetup{
  colorlinks=true,
  linkcolor=blue,
  citecolor=blue,
  filecolor=magenta,
  urlcolor=cyan}
\usepackage[titletoc]{appendix}
\usepackage{float,bm}

\makeatletter
\def\rf#1{(\@rf#1,.)}
\def\@rf#1,{\ref{eq:#1}\@ifnextchar . {\@endrf}{, \@rf}}
\def\@endrf.{}
\makeatother

\newcommand\qaq{\quad\text{and}\quad}

\newcommand\defterm[1]{\cite[#1]{tgn2015c}}

\renewcommand{\~}[1]{\mathbf{#1}}
\newcommand{\setbg}[1]{\bigl\{\,#1\,\bigr\}}
\newcommand{\casemod}[2]{\left\{ \begin{array}{#1} #2 \end{array}\right. }

\newcommand{\pendL}{L}
\newcommand{\pendG}{g}

\newcommand{\Sigmeth}{$\Sigma$-method\xspace}

\newcommand{\s}[1]{\text{\scriptsize$#1$}}

\newcommand\offsetpair{(\~c;\~d)}
\newcommand\ssth[1]{#1{\text{th}}}

\newcommand{\sij}[2]{\sigma_{#1#2}}

\newcommand{\Sig}{\Sigma}
\newcommand{\Jac}{\~J}
\newcommand{\Jij}[2]{\Jac_{#1#2}}

\newcommand\bninfty{\boldsymbol{-\infty}}

\newcommand{\newsig}{\overline\sigma}
\newcommand{\newsij}[2]{\newsig_{#1#2}}

\newcommand{\newSig}{\overline\Sigma}

\newcommand{\newJ}{\~{\overline J}}
\newcommand{\newJac}{\~{\overline J}}

\newcommand{\newT}{\overline T}
\newcommand{\newf}{\overline f}

\newcommand{\newc}{\overline c}
\newcommand{\newd}{\overline d}

\newcommand{\Aess}{\~A_{\text{ess}}}
\newcommand{\Sess}{S_{\text{ess}}}
\newcommand{\rnge}[2]{#1\,{:}\,#2}

\newcommand{\daesa}{{\sc daesa}\xspace}
\newcommand{\matlab}{{\sc Matlab}\xspace}

\newcommand{\lam}{\lambda}

\newcommand{\indxk}{l}
\newcommand{\nzset}{L}
\newcommand{\eqsetI}{I}

\newcommand{\hoder}[2]{\sigma\left(#1,#2\right)}

\newcommand\eqsetIES{M}

\newcommand\nzsetES{J}

\newcommand{\jtwo}{r}

\newcommand{\val}[1]{\textnormal{Val}(#1)}
\newcommand{\set}[1]{\{\, #1 \,\}}
\newcommand{\pp}[2]{\frac{\partial #1}{\partial #2}}
\newcommand{\ppin}[2]{\partial #1/\partial #2}
\newcommand{\fracin}[2]{#1/#2}

\newcommand{\blk}[1]{\text{blockOf}(#1)}

\definecolor{lgray}{rgb}{0.6,0.6,0.6}

\definecolor{palegray}{rgb}{0.7,0.7,0.7}
\newcommand{\OK}[1]{\colorbox{palegray}{\ensuremath{#1}}}

\newcommand{\rrf}[2]{(\ref{eq:#1}--\ref{eq:#2})}

\newcommand{\scref}[1]{\ref{sc:#1}}

\newcommand{\thref}[1]{\ref{th:#1}}

\newcommand{\SSCref}[1]{\S\ref{ssc:#1}}

\newcommand{\SCref}[1]{\S\ref{sc:#1}}
\newcommand{\EXref}[1]{Example~\ref{ex:#1}}

\newcommand{\LEref}[1]{Lemma~\ref{le:#1}}
\newcommand{\FGref}[1]{Figure~\ref{fg:#1}}

\newcommand{\THref}[1]{Theorem~\ref{th:#1}}

\newcommand\nequiv{\not\equiv}

\newcommand\idq{{q}}
\newcommand\idQ{{Q}}
\newcommand\idv{{w}}
\newcommand\idr{{q}}
\newcommand{\llra}{\Longleftrightarrow}

\newcommand{\wblk}[1]{\overline{\text{blockOf}}(#1)}

\newcommand{\wN}[1]{\overline N_{#1}}

\newcommand{\lbar}[1]{\multicolumn{1}{|c}{$#1$}}
\newcommand{\dbar}{\lbar{}}
\newcommand{\SpJac}{\~S_0}

\newcommand{\hu}{{\widehat u}}
\newcommand{\hv}{{\widehat v}}

\newcommand\phanbull{\phantom{^\bullet}}
\newcommand\LCconst{\underline{c}}
\newcommand\ESconst{\overline{c}}
\newcommand\underC{\LCconst}
\newcommand\overC{\ESconst}
\newcommand\vecu{\~u}
\newcommand\vecv{\~v}
\newcommand\newL{{\overline \nzset}}
\newcommand\newJes{{\overline \nzsetES}}
\newcommand\Jblk[2]{\Jij{[#1}{#2]}}
\newcommand\newJblk[2]{\overline\Jac_{#1#2}}

\renewcommand\ker[1]{\textnormal{ker}(#1)}
\newcommand\coker[1]{\textnormal{coker}(#1)}

\renewcommand\SCref[1]{Section~\ref{sc:#1}}
\renewcommand\newc{{\widetilde c}}
\renewcommand\newd{{\widetilde d}}
\renewcommand\Sig{\boldsymbol{\Sigma}}
\renewcommand{\newSig}{\overline{\boldsymbol{\Sigma}}}

\newcommand\Jnol{\nzsetES\setminus\setbg\indxk}
\renewcommand\Jblk[2]{\Jac_{#1#2}}
\renewcommand\Jij[2]{J_{#1#2}}

\renewcommand\idr{r}
\renewcommand\idv{m}
\renewcommand\SpJac{\~S_0\offsetpair}
\newcommand\cidjEqnum{(2.2)}
\newcommand\sysjacEqnum{(2.6)}
\newcommand\truehodEqnum{(4.1)}
\newcommand\LCdefEqnum{(4.3)}
\newcommand\LCcondEqnum{(4.4)}
\newcommand\ESdefEqnum{(4.10)}
\newcommand\EScondEqnum{(4.11)}

\journalname{BIT}

\begin{document}
\title{Conversion Methods, Block Triangularization, and Structural Analysis of Differential-Algebraic Equation Systems}

\titlerunning{Conversion Methods, Block Triangularization, and Structural Analysis of DAEs}

\author{Guangning Tan         \and
Nedialko S. Nedialkov \and\\
John D. Pryce
}


\institute{G. Tan
\at School of Computational Science and Engineering, McMaster University, \at
1280 Main Street West, Hamilton, Ontario L8S 4K1, Canada,
\email{tang4@mcmaster.ca} \and
N. S. Nedialkov
\at Department of Computing and Software, McMaster University, \at
1280 Main Street West, , Hamilton, Ontario L8S 4K1, Canada,
\email{nedialk@mcmaster.ca} \and
J. D. Pryce
\at School of Mathematics, Cardiff University, \at Senghennydd Road, Cardiff CF24 4AG, Wales, UK.,
\email{prycejd1@cardiff.ac.uk}
}

\date{Received: date / Accepted: date}

\maketitle

\begin{abstract}
In a previous article, the authors developed two conversion methods to improve the \Sigmeth for structural analysis (SA) of differential-algebraic equations (DAEs). These methods reformulate a DAE on which the \Sigmeth fails into an equivalent problem on which this SA is more likely to succeed with a generically nonsingular Jacobian. 
The basic version of these methods processes the DAE as a whole.
This article presents the block version that exploits block triangularization of a DAE. Using a block triangular form of a Jacobian sparsity pattern, we identify which diagonal blocks of the Jacobian are identically singular and then perform a conversion on each such block. This approach improves the efficiency of finding a suitable conversion for fixing SA's failures. All of our conversion methods can be implemented in a computer algebra system so that every conversion can be automated.
\keywords{differential-algebraic equations \and structural analysis \and block triangular form \and modeling\and symbolic computation}
\subclass{34A09\and 65L80\and 41A58 \and 68W30}
\end{abstract}


\def\bitinfo{\vbox{\hbox to 12.2cm{}
       \hbox to 12.2cm{}
      \hfill}} 
\setcounter{page}{001}

\section{Introduction.}\label{sc:intro}

This article is a continuation of \cite{tgn2015c}, 
in which we presented two conversion methods for improving the \Sigmeth \cite{pryce2001a} for structural analysis (SA) of DAEs.  When this SA fails on a DAE with an identically singular (but structurally nonsingular) System Jacobian, our conversion methods reformulate the DAE into an equivalent problem on which the SA is more likely to succeed with a generically nonsingular System Jacobian \cite{tgn2015aCAS,tgn2015c}.

These two conversion methods are the linear combination (LC) method and the expression substitution (ES) method. The former is based on replacing an existing equation by a linear combination of some equations and derivatives of them. The latter is based on replacing some existing derivatives\footnote{Throughout this article, ``derivatives of a variable'' include the variable itself as its $0$th derivative.} by expressions that contain newly introduced variables and derivatives of them. In the ES method, the equations that prescribe such replacements are also appended to the original DAE, so the resulting system is an enlarged one. The main result of a conversion using either method is a strict decrease in the value of the signature matrix \cite{tgn2015c}. Based on our experience, we conjecture that such a decrease tends to give a better problem formulation of a DAE from SA perspective.

Our works \cite{NedialkovPryce2012a, NedialkovPryce2012b, pryce2014btf} show how to construct block triangular forms (BTFs) of a DAE using the structural data obtained from the \Sigmeth. A BTF indicates how each part of the DAE influences [resp. is influenced by] other parts. The interdependences between all pairs of blocks may be depicted by a fine-block graph \cite{pryce2014btf}. Exploiting the underlying structure of a DAE, we can compute the derivatives of its solution in a blockwise fashion \cite{nedialkov2016a}, or perform a dummy derivative index reduction algorithm \cite{mckenzie2015a, McKenzie2016a}. 

We refer the reader to the previous article \cite{tgn2015c} for a summary of the \Sigmeth, details of its failures, the basic conversion methods, and explanations of the equivalence of DAEs. By ``basic'' we mean that these methods do not exploit BTFs of a DAE. We shall follow the notation in \cite{tgn2015c}. 

In this article, we combine our conversion methods with a block triangularization of a DAE and derive our block conversion methods. When the System Jacobian is identically singular, and the DAE has a nontrivial BTF---that is, having at least two diagonal blocks---we can identify which blocks are identically singular and perform a conversion on each such block. Now that we only deal with equations and variables within a block, which is usually of a smaller size compared to the whole DAE, these block methods require fewer symbolic computations and hence are expected to be more efficient in finding a useful conversion for fixing SA's failures.

\SCref{SA} reviews BTFs of a sparsity pattern and BTFs of a DAE.
\SCref{convBTF} presents our block conversion methods and demonstrates their application on a DAE from \cite{campbell1995solvability}. \SCref{exam} gives more examples, in which the two DAEs are obtained from electrical circuit analysis \cite{TestSetIVP}. \SCref{conclu} gives concluding remarks.

\renewcommand{\SCref}[1]{\S\ref{sc:#1}}

\section{Block triangularization of DAEs.}\label{sc:SA}
In \SSCref{BTF}, we introduce notation for a BTF of a sparsity pattern.
In \SSCref{BTFDAE}, we review how to derive a BTF of a DAE; more details are in \cite{NedialkovPryce2012a, pryce2014btf}.

We do not repeat the definitions and formulas for the notation in the \Sigmeth theory, such as a {\em signature matrix} $\Sig=(\sij{i}{j})$ and its {\em value} $\val\Sig$, a {\em highest-value transversal (HVT)} $T$ of $\Sig$, a {\em valid offset pair} $\offsetpair$, a {\em System Jacobian} $\Jac\offsetpair=(\Jij{i}{j})$, and so forth. We refer the reader to \cite{tgn2015c} for details.

Terms are in {\sl slanted font} at their defining occurrence. We use bold font for matrices that may split into blocks, and for the sub-matrices. Individual entries of a matrix are in lowercase. For example, matrix $\~A$ has sub-matrices $\~A_{lm}$ and entries $a_{ij}$.

\subsection{Block triangular forms of a sparsity pattern.}\label{ssc:BTF}
Let\footnote{The colon notation $\rnge{p}{q}$ for integers $p, q$ denotes either the unordered set or the enumerated list of integers $i$ with $p\le i\le q$, depending on context.} 
$R=\rnge{1}{n}$ be the set of indices of $n$ rows (equations), and let  $C=\rnge{1}{n}$ be the set of indices of $n$ columns (variables). 
A {\sl sparsity pattern} $\~A$ is a subset of the Cartesian product $R\times C$ that contains row-column index pairs $(i,j)$. We can view $\~A$ as its incidence matrix $(a_{ij})$, where $a_{ij}$ equals 1 if $(i,j)\in\~A$ and 0 otherwise. A {\sl transversal} of $\~A$ is $n$ positions in $\~A$ with exactly one position in each row and each column. If $\~A$ has some transversal, then it is {\sl structurally nonsingular}.  
The union of all transversals of $\~A$ comprise its {\sl essential sparsity pattern} $\Aess$ \cite{pryce2014btf}. Obviously, $\~A$ is structurally nonsingular if and only if $\Aess$ is nonempty.

Assume henceforth that $\~A$ is structurally nonsingular. Let $P$ and $Q$ be two suitable permutation matrices for $\~A$, such that the permuted incidence matrix $\~A'=P\~AQ$ can be written in a $p\times p$ block form
\begin{align}\label{eq:btf0}
\~A' =
\left[\begin{array}{*4{@{\;}c}@{}}
\~A_{11} &\~A_{12} &\cdots &\~A_{1p} \\
&\~A_{22} &\cdots &\~A_{2p} \\
&         &\ddots &\vdots \\
&         &       &\~A_{pp}
\end{array}\right]\;,
\end{align} 
where each diagonal block $\~A_{\idq\idq}$, $\idq=\rnge{1}{p}$, is square of positive size $N_\idq$. 
We say the block form \rf{btf0} is a {\sl BTF} of $\~A$. 
Blanks in \rf{btf0} mean that a sub-matrix $\~A_{kl}$ below the block diagonal with $k>l$ is empty.

A sparsity pattern is {\sl irreducible}, if it cannot be permuted to the form \rf{btf0} with $p>1$ \cite{Duff86a}; otherwise it is {\sl reducible}. A BTF is  {\sl irreducible} if each diagonal block is {\sl irreducible}; otherwise it is {\sl reducible} \cite{pryce2014btf}. Hence, if \rf{btf0} is  irreducible, then $p$ is the largest number of diagonal blocks among all possible BTFs of $\~A'$.

When we say block $\idq$ of a matrix in a BTF, we shall refer to the $\ssth{\idq}$ diagonal block submatrix. For $\idq=\rnge{1}{p}$, we define for block $\idq$ the index set
\begin{equation*}
\begin{aligned}
B_\idq &= \text{the set of indices $i$ that belong to block $\idq$}  
\;.
\end{aligned}
\end{equation*}
Throughout this article, they are the indices of the permuted $\~A'$, not those of the original $\~A$.

Another useful notation is $\blk{i}$ that denotes the block number $\idq$ such that index $i\in B_\idq$. Since each diagonal block is square, both $B_\idq$ and $\blk{i}$ notation apply to rows and columns equally. To summarize, for $i\in \rnge{1}{n}$ and $\idq\in\rnge{1}{p}$,
\begin{align*}
\blk{i} = \idq 
\llra i\in B_\idq 
&\llra  \sum_{\idv=1}^{\idq-1} N_\idv+1 \le i \le\sum_{\idv=1}^{\idq} N_\idv\;.
\end{align*}

\begin{example}
We illustrate in \FGref{btfA} the above block notation with a sparsity pattern of two nontrivial BTFs. 
\begin{figure}
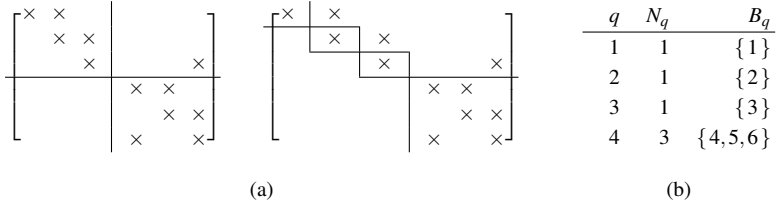

\begin{tabular}{ccc}
$ 
 \begin{blockarray}{*6{c@{\hspace{5pt}}}}\\
\begin{block}{[*6{c@{\hspace{5pt}}}]}
\times & \times   &  &\dbar  &   & \\
& \times  & \times &\dbar&   &\\  
&   & \times   & \dbar  & &  \times\\ \cline{1-6}
&&& \lbar{\times } & \times & \\
&&& \dbar  & \times & \times  \\
&&& \lbar{\times } &   & \times\\
\end{block}
\end{blockarray}
\qquad
 \begin{blockarray}{*6{c@{\hspace{5pt}}}}\\
\begin{block}{[*6{c@{\hspace{5pt}}}]}
\times & \lbar{\times}   &  &  &   & \\\cline{1-2}
& \lbar{ \times } & \lbar{ \times } &&   &\\   \cline{2-3}
&   &\lbar{\times}   & \dbar  & &  \times\\ \cline{3-6}
&&& \lbar{\times } & \times & \\
&&& \dbar  & \times & \times\\
&&& \lbar{\times } &   & \times\\
\end{block}
\end{blockarray}
$
&\quad&  $\def\arraystretch{1.2}\begin{array}{*3{@{\hspace{10pt}}r}}
 \idq &  N_\idq  & B_\idq  \\ \hline
1 & 1 & \set{1}  \\
2 & 1 & \set{2}  \\
3 & 1 & \set{3}  \\
4 & 3 & \set{4,5,6}
  \end{array} $
  \\ (a) &\quad& (b)
\end{tabular}  
\caption{
\label{fg:btfA}
(a) Two nontrivial BTFs of the same sparsity pattern. The left one is reducible with number of blocks $p=2$. The right one is irreducible with $p=4$. 
(b) Block information for the irreducible BTF.
}
\end{figure}
\end{example}

The following lemma connects the transversals of a sparsity pattern $\~A$ and the transversals of its diagonal blocks in some BTF.
\begin{lemma}\label{le:BTF4.1}\cite[Lemma 2.4]{pryce2014btf}
Any transversal $T$ of a sparsity pattern $\~A$ is contained in the union of the diagonal blocks of any BTF of $\~A$, that is, $T\subseteq \~A_{11}\cup\cdots\cup\~A_{pp}$.
\end{lemma}

Equivalently, the intersection of $T$ with block $\idq$ of $\~A$ is a transversal $T_\idq$ of $\~A_{\idq\idq}$.

\subsection{Block triangular forms of a DAE.}\label{ssc:BTFDAE}
The natural sparsity pattern of a DAE indicates if a variable $x_j$ occurs in an equation $f_i$. Each such occurrence corresponds to a finite entry $\sij{i}{j}$ in $\Sig$, and hence we have
\begin{equation*}
\~S= \setbg{(i,j)\ \mid\ \sij{i}{j}>-\infty\ }\qquad\text{(the sparsity pattern of $\Sig$)}\;.
\end{equation*}
If $\~S$ has some transversal, then $\Sig$ has a transversal with finite $\sij{i}{j}$'s and a finite $\val\Sig$ \cite{pryce2001a}, so the DAE is structurally well posed (SWP) \cite{NedialkovPryce2012a}; otherwise it is structurally ill posed. Here, we shall deal with the SWP case only.

A more informative BTF derives from the sparsity pattern $\~S_0 = \SpJac$ of a System Jacobian $\Jac=\Jac\offsetpair$ as defined in \defterm{\sysjacEqnum}:
\renewcommand\SpJac{\~S_0}
\begin{equation}\label{eq:spJac}
\SpJac= \setbg{(i,j)\ \mid\ d_j-c_i=\sij{i}{j}\ }\qquad\text{(the sparsity pattern of $\Jac$)}\;.
\end{equation}
By \defterm{\cidjEqnum}, $d_j-c_i=\sij{i}{j}$ holds on a HVT $T$ of $\Sig$, so $T$ is also a transversal of $\SpJac$.

A less obvious set contains the positions that contribute to $\det(\Jac)$:
\begin{equation*}
\~\Sess = \text{the union of all HVTs of $\Sig$}\qquad\text{(the essential sparsity pattern of $\Sig$)}\;,
\end{equation*}
which is also the essential sparsity pattern of $\SpJac$ for any  valid offset pair $\offsetpair$ \cite[Lemma 3.1]{pryce2014btf}.

Since $d_j-c_i=\sij{i}{j}$ holds on each HVT and hence implies $\sij{i}{j}>-\infty$, we have
\begin{equation*}
\~\Sess \subseteq \SpJac \subseteq \~S \qquad\text{for any offset pair $\offsetpair$}\;.
\end{equation*}
Our experience suggests that the (irreducible) BTF based on $\SpJac$ can be significantly finer than that based on $\~S$. We refer to the former BTF as {\sl fine BTF}, and to the latter as {\sl coarse BTF}. We refer to the diagonal blocks in the fine BTF as {\sl fine blocks}, and refer to those in the coarse BTF as {\sl coarse blocks}.

Assume that $\SpJac$ is permuted into a $p\times p$ BTF. Following this BTF, we apply the same permutations on $\Jac$ and $\Sig$, and write them  in $p\times p$ block forms:
\begin{align}\label{eq:pBypBlockForm}
\Jac = 
\begin{blockarray}{*5c}
\begin{block}{[*5{@{\hspace{8pt}}c}]}
 \Jblk{1}{1} &\Jblk{1}{2} &\cdots &\Jblk{1}{p} \\
 \~0 &\Jblk{2}{2} &\cdots &\Jblk{2}{p} \\
 \vdots      &\ddots    &\ddots  &\vdots \\
 \~0  &\cdots   &\~0   &\Jblk{p}{p} \\
\end{block}
\end{blockarray} 
\qaq
\Sig = 
\begin{blockarray}{*5c}
\begin{block}{[*5{@{\hspace{8pt}}c}]}
\Sig_{11}&\Sig_{12} &\cdots &\Sig_{1p} \\
\Sig_{21}      &\Sig_{22} &\cdots &\Sig_{2p} \\
\vdots      &\ddots       &\ddots  &\vdots \\
\Sig_{p1}     &\cdots      &\cdots   &\Sig_{pp}\\
\end{block}
\end{blockarray}
\;.
\end{align}
We call this procedure a {\sl block triangularization} of the DAE, and note that  the sparsity pattern $\~S$ of $\Sig$ may not be in the same BTF as $\SpJac$ of $\Jac$. That is, every $\sij{i}{j}$  below the block diagonal of $\Sig$ is not necessarily $-\infty$, but must satisfy $\sij{i}{j}<d_j-c_i$ as $\Jij{i}{j}\equiv 0$. Hence,
\begin{align}\label{eq:djciblk}
d_j - c_i
\left\{
\begin{alignedat}{7}
&>\sij{i}{j}     &\quad& \text{if $\blk{j} < \blk{i}$}\\
&\ge \sij{i}{j}  && \text{if $\blk{j} \ge \blk{i}$}\;.
\end{alignedat}
\right.
\end{align}
We refer the reader to \cite{NedialkovPryce2012b,pryce2014btf} for more details on BTFs.

\begin{example}\label{ex:mod2pend}
We illustrate the coarse and fine BTFs with the (artificially) modified double pendula DAE in \cite{nedialkov2016a}. The state variables are $x,y,\lam,u,v,\mu$; $G$ is gravity, $L>0$ is the length of both pendula, and $\alpha$ is  a constant. 
\begin{equation*}
\begin{alignedat}{5}
0 = f_1 &= x''+x\lam
&\qquad & 0 = f_4 &&= u''+u\mu \\
0 = f_2 &= y''+y\lam+(x')^3 -\pendG
&& 0 = f_5 &&= (v''')^3+v\mu-G\\
0 = f_3 &= x^2+y^2-\pendL^2 
&& 0 = f_6 &&= u^2+v^2-(L+\alpha\lambda)^{2}+\lambda''
\end{alignedat}
\end{equation*}
\begin{align*}
\Sig=
\begin{blockarray}{r@{\hspace{1mm}}*6{c@{\hspace{2mm}}}c@{\hspace{2mm}}}
       &  v\phanbull &   \mu &  u &  x &  y &  \lam & \s{c_i} \\
\begin{block}{r@{\hspace{1mm}}[*6{c@{\hspace{2mm}}}]c}
f_5\;\;\;\; & 3^\bullet  & \lbar{0\phanbull}  &  & \dbar &   &    &   \s0 \\\cline{2-3}
f_4\;\;\;\; &    & \lbar{\ 0^\bullet\ }   & \lbar{2\phanbull} &  \dbar &  &    &   \s0\\ \cline{3-4}
f_6\;\;\;\; & \OK{0}\phanbull &  &\lbar{\ 0^\bullet\ }    & \dbar &  &   2\phanbull & \s2\\ \cline{2-7}
f_3\;\;\;\; &    &   &    &\lbar{0^\bullet} & \; 0\phanbull &   &   \s6\\
f_2\;\;\;\; &    &   &      & \lbar{\OK{1}\phanbull} &\; 2^\bullet   & 0\phanbull & \s4\\
f_1\;\;\;\; &    &   &    &\lbar{2\phanbull} & &0^\bullet    & \s4\\
\end{block}
\s{d_j}& \s3 &\s0 &\s2 &\s6 &\s6 &\s4 \\[-5ex]
\end{blockarray}
\quad
\Jac=
\begin{blockarray}{r@{\hspace{2mm}}*6{c@{\hspace{1mm}}}}
\\
 &  v''' &   \phanbull \mu &  \phanbull u'' &  \phanbull x^{(6)} &  y^{(6)} &  \lam^{(4)} \\
\begin{block}{l@{\hspace{2mm}}[*6{c@{\hspace{1mm}}}]}
f_5 & 2v''' & \lbar{v}   &  &&   & \\\cline{2-3}
f_4 &  & \lbar{u} & \lbar{1} & &   &
\\   \cline{3-4}
f''_6 &  &   &\lbar{2u}   & \dbar  & &  1
\\ \cline{4-7}
f_3^{(6)} &&&& \lbar{2x} & 2y & \\
f_2^{(4)} &&&& \dbar  & 1 & y\\
f_1^{(4)} &&&& \lbar{1} &   & x\\
\end{block}
\end{blockarray}
\end{align*}
The row and column labels in $\Jac$, showing equations and variables differentiated to order $c_i$ and $d_j$, aim to remind the reader of the formula for $\Jac$ in \defterm{\sysjacEqnum}.

There are two $3\times 3$ coarse blocks. The first one, comprising equations $f_5, f_4, f_6$ and variables $v, \mu, u$, can further decompose into three $1\times 1$ fine blocks, while the second coarse block, comprising equations $f_3, f_2, f_1$ and variables $x,y,\lam$, is irreducible. Hence there are four blocks in the fine BTF.

The sparsity pattern $\SpJac$ of $\~J$ is exactly the one in \FGref{btfA}(a), so the fine BTF information is in \FGref{btfA}(b).
\end{example}

If we state \LEref{BTF4.1} in the context of a Jacobian sparsity pattern, then we have the following lemma.
\begin{lemma}\label{le:unionHVT}\cite[Lemma 3.3]{pryce2014btf}
Assume that a Jacobian sparsity pattern $\SpJac$ is in some BTF. Let $(\Sig_{\idq m})_{\idq,m=\rnge{1}{p}}$ be the corresponding sub-matrices of $\Sig$. Then a HVT $T$ of $\Sig$ is the union of HVTs $T_q$ of the diagonal blocks $\Sig_{\idq\idq}$: $T=T_1\cup\cdots\cup T_p$.
\end{lemma}

This lemma is not difficult to prove, given that a transversal $T$ of   $\SpJac$ is the union of transversals $T_\idq$ of the diagonal blocks of $\SpJac$.

The following lemma is useful for proving the main Theorems~\thref{LCblk} and \thref{ESblk} of the block conversion methods in \SCref{convBTF}.
\begin{lemma}\label{le:blockstructure}
Assume that $\Sig$ has a finite $\val\Sig$ and is in a $p\times p$ block form as in \rf{pBypBlockForm}.
Let $\~c$ and $\~d$ be two nonnegative integer $n$-vectors. Assume also that
\begin{enumerate}[(a)]
\item $d_j-c_i>\sij{i}{j}$ holds for all entries below the diagonal blocks of $\Sig$, 
\item $d_j-c_i\ge \sij{i}{j}$ holds elsewhere, and 
\item $\val\Sig=\sum_j d_j - \sum_i c_i$.
\end{enumerate}
Then 
\begin{enumerate}[(i)]
\item $\offsetpair$ is a valid offset pair of $\Sig$,
\item the block form of $\Sig$ is a BTF of the Jacobian sparsity pattern $\SpJac$, and
\item a HVT of $\Sig$ is the union of HVTs $T_q$ of the diagonal blocks $\Sig_{\idq\idq}$, for all $\idq=\rnge{1}{p}$.
\end{enumerate}
\end{lemma}
\begin{proof}
{\em (i)} We let $T$ denote a HVT of $\Sig$. Since $\val\Sig$ is finite, $\sij{i}{j}\ge0$ for all $(i,j)\in T$.
For $\offsetpair$ to be a valid offset of $\Sig$, $d_j-c_i\ge\sij{i}{j}$ must hold for all $i,j=\rnge{1}{n}$, with equalities for all $(i,j)\in T$ \cite{pryce2001a}.

By (a) and (b), $d_j-c_i\ge \sij{i}{j}$ holds everywhere. Summing these inequalities over $T$ gives 
\[
\sum_{(i,j)\in T} (d_j-c_i)  \ge \sum_{(i,j)\in T} \sij{i}{j}\;.
\]
The left-hand side equals $\sum_j d_j-\sum_i c_i$, and the right-hand side equals $\val\Sig$ by definition. By (c), these two values are equal,  so  $d_j-c_i=\sij{i}{j}$ holds for all $(i,j)\in T$, and $\offsetpair$ is valid for $\Sig$.

{\em (ii)} By (a), the blocks below the block diagonal in $\SpJac$, derived from $\Sig$ and $\offsetpair$ using \rf{spJac}, are empty. By the definition of a BTF of a Jacobian sparsity pattern, $\SpJac$ is in a BTF as described by the $p\times p$ block form.

{\em (iii)} This follows immediately from {\em (ii)} and \LEref{unionHVT}. \qed
\end{proof}

Following a $p\times p$ BTF based on $\SpJac$, we can write any valid offset pair $\offsetpair$ of $\Sig$ in a block form as
\begin{align}\label{eq:offsetblockform}
(\~c_1; \~d_1), (\~c_2; \~d_2), \ldots, (\~c_p; \~d_p)\;,
\end{align}
where each of the sub-vectors $\~c_\idq$ and $\~d_\idq$ is of length $N_\idq$, where $\idq=\rnge{1}{p}$. 

\pagebreak

\begin{lemma}\label{le:BTF4.2}
Assume that a Jacobian pattern $\SpJac$, derived by $\Sig$ and a valid offset pair $\offsetpair$, is in some BTF. If we write $\offsetpair$ into block form as in \rf{offsetblockform}, then $(\~c_\idq;\~d_\idq)$ is a valid offset pair of $\Sig_{\idq\idq}$.
\end{lemma}

\begin{proof}
Let $T$ be a HVT of $\Sig$. By \LEref{unionHVT}, the intersection of $T$ with block $\idq$ is a HVT $T_\idq$ of $\Sig_{\idq\idq}$. Then $d_j-c_i=\sij{i}{j}$ holds for all  $(i,j)\in T_\idq\subseteq T$. Since $\offsetpair$ is valid for $\Sig$, $d_j-c_i\ge\sij{i}{j}$ and $c_i\ge0$ hold on $\Sig_{\idq\idq}$, where $i,j \in B_\idq$. Thus the offset pair $(\~c_\idq;\~d_\idq)$ matched to block $\idq$ satisfies the conditions \defterm{\cidjEqnum} for being valid for $\Sig_{\idq\idq}$. \qed
\end{proof}

From the view of \LEref{BTF4.2}, we can regard each diagonal block $\Sig_{\idq\idq}$ as a signature matrix in its own right. Equivalently, each block $\idq$, having $N_\idq$ equations in $N_\idq$ variables, can be viewed as a sub-DAE, with a signature matrix $\Sig_{\idq\idq}$, a finite $\val{\Sig_{\idq\idq}}$, a {\sl local offset pair} $(\~c_\idq; \~d_\idq)$, and a {\sl sub-Jacobian} $\Jblk{\idq}{\idq}$. Expressions that contribute to entries in an off-diagonal block $\Sig_{\idq m}$, where $\idq\neq m$, can be considered as driving terms, or equivalently, the influence of variables in block $m$ on those in block $\idq$.
We refer to $\offsetpair$ of $\Sig$ as a {\sl global offset pair}. The reader is referred to \cite{pryce2014btf} for more theoretical results about block triangularization and global/local offset pairs.

\section{Block conversion methods.}\label{sc:convBTF}

They are suitable for improving the efficiency of finding a useful conversion for fixing SA's failures. If $\Jac$ is identically singular, then by \rf{pBypBlockForm}, $\det(\Jac)=\prod_{\idq=1}^p \det(\Jac_{\idq\idq})\equiv 0$, so at least one $\Jac_{\idq\idq}$ for some $\idq\in\rnge{1}{p}$ is identically singular. As discussed before, we can regard block $\idq$ as a sub-DAE with a signature matrix $\Sig_{\idq\idq}$. Then we wish to apply the basic conversion methods on this sub-DAE to achieve a strict decrease in $
\val{\Sig_{\idq\idq}}$, provided the conditions for applying these methods are satisfied for those variables and equations within block $\idq$.

However, what we should ensure is a strict decrease in the value of the {\em whole} signature matrix, namely $\val\newSig<\val\Sig$, where $\newSig$ is the signature matrix of the resulting DAE. Proving this inequality  from a decrease in $\val{\Sig_{\idq\idq}}$ is nontrivial, because a conversion on block $\idq$ may affect blocks $\Sig_{\idq m}$ for $m=1,\ldots,\idq-1,\idq+1,\ldots,p$. Especially in the ES method, $\Sig_{\idq\idq}$ and these blocks are enlarged. Hence, the conditions and the conversion process need to be carefully modified, so that the conversion methods can adapt to a BTF based on $\SpJac$.

We give an introductory example in \SSCref{intro}, present the block LC method in \SSCref{LCBTF}, and present the block ES method in  \SSCref{ESBTF}.

Hereafter we use the fine BTF in the examples for demonstration, since each fine block contains an irreducible sub-Jacobian sparsity pattern. Our experience suggests that a useful conversion can usually be derived from the fine BTF of a DAE.  However, we emphasize that the block conversion methods can be applied not only  to the irreducible BTF of a Jacobian sparsity pattern $\SpJac$ with some valid $\offsetpair$, but also to any BTF of $\SpJac$. For example, the basic conversion methods consider a DAE in a (trivial) BTF of one $n\times n$ block.

\subsection{An introductory example.}\label{ssc:intro}
We illustrate these block methods with the following DAE:
\begin{equation}\label{eq:blockLC1}
\begin{aligned}
0 &= f_1 = x_1+x_2 + h_1(t)\\
0 &= f_2 = x_1 + (x'_1+x'_2)x'_3 + h_2(t) \\
0 &= f_3 = x'_3 + h_3(t)\;.
\end{aligned}
\end{equation}
\begin{align*}
\Sig = 
\begin{blockarray}{rc@{\hspace{2mm}}c@{\hspace{2mm}}c@{\hspace{10pt}}cc}
&  x_1\phanbull &   x_2\phanbull & x_3  & \s{c_i} \\
\begin{block}{r @{\hspace{10pt}}[c@{\hspace{2mm}}c|@{\hspace{2mm}}c@{\hspace{2mm}}]cc}
f_1 & 0^\bullet  &0\phanbull  & \phantom{-}  &   \s1  \\
f_2 & 1\phanbull  &1^\bullet  & 1\phanbull &   \s0  \\ \cline{2-4}
f_3 &  & & 1^\bullet & \s0  \\
\end{block}
\s{d_j}& \s1\phanbull &\s1\phanbull &\s1
\end{blockarray}
\qquad
\Jac = 
\begin{blockarray}{rc@{\hspace{2mm}}cc@{\hspace{10pt}}cc}
&  x'_1\phanbull &  x'_2 & x'_3  &  \\
\begin{block}{r @{\hspace{10pt}}[c@{\hspace{2mm}}c|c@{\hspace{2mm}}]cc}
f'_1 & 1  &1  & &     \\
f_2 & x'_3  &x'_3  & x_1'+x_2' &    \\ \cline{2-4}
f_3 &  &  & 1 &   \\
\end{block}
&
\end{blockarray}
\end{align*}
(Here $h_1$, $h_2$, $h_3$ are driving functions.) The coarse BTF and the fine BTF are identical, both having two diagonal blocks.

In the basic LC method, we can choose $\vecu=[-x'_3,1,-x'_1-x'_2]^T\in\coker\Jac$. Using \defterm{\LCdefEqnum}, we have
\begin{align*}
\eqsetI = \setbg{i\mid u_i\nequiv 0} = \setbg{1,2,3},
\quad 
\underC=\min_{i\in\eqsetI} c_i = 0, \quad
\nzset = \setbg{\indxk\in\eqsetI\mid c_\indxk=\underC}=\setbg{2,3}\;.
\end{align*}
We let $\hoder{x_j}{\vecu}$ denote the order of the highest derivative to which $x_j$ occurs in $\vecu$, or $-\infty$ if $x_j$ does not occur in $\vecu$ \defterm{\truehodEqnum}. The LC condition \defterm{\LCcondEqnum} is violated since
\[
\hoder{x_j}{\vecu}=1\not< 1=d_j-\underC \quad \text{for all}\quad j=\rnge{1}{3}\;. 
\]
Not surprisingly, replacing either  $f_2$ or $f_3$ by
\[
\newf = \sum_{i\in \eqsetI} u_i f_i^{(c_i-\underC)} 
= -x'_3 h'_1(t) + \bigl(x_1 + h_2(t)\bigr) - (x'_1+x'_2)\bigl(x'_3 + h_3(t)\bigr)
\]
does not result in a decrease in $\val\Sig$; verifying this is not difficult.

Notice that only the sub-Jacobian of block 1, $\Jblk{1}{1}=\ppin{(f'_1,f_2)}{(x'_1,x'_2)}$, is singular. Suppose we consider block 1, with $B_1=\setbg{1,2}$, as a sub-DAE, and choose $\vecu=[-x'_3,1]^T \in\coker{\Jblk{1}{1}}$. Within block 1, the LC method derives
\begin{align*}
\eqsetI = \setbg{i\in B_1 \mid u_i\nequiv 0} = \setbg{1,2},
\quad 
\underC=\min_{i\in\eqsetI} c_i = 0, 
\quad
\nzset = \setbg{\indxk\in\eqsetI\mid c_\indxk=\underC}=\setbg{2}\;.
\end{align*}
Now the LC condition \defterm{\LCcondEqnum} is satisfied for the column indices in block 1:
\[
\hoder{x_j}{\vecu}=-\infty<d_j-\underC \quad\text{for $j=1,2\in B_1$}\;.
\] 
Replacing $f_2$ by $\newf_2 = u_1 f'_1 + u_2 f_2= x_1 + h_2(t) - x'_3 h'_1(t)$ results in the DAE with the following SA result.

\begin{equation*}
\newSig = 
\begin{blockarray}{r@{\hspace{3mm}}ccccc}
&  x_2\phanbull &   x_1 & x_3  & \s{c_i} \\
\begin{block}{r @{\hspace{10pt}}[ccc]cc}
f_1 & 0^\bullet  &\lbar{0\phanbull}  & \phantom{-}  &   \s0  \\ \cline{2-3}
\newf_2 &   &\lbar{0^\bullet}  & \lbar{1\phanbull} &   \s0  \\ \cline{3-4}
f_3 &  & & \lbar{1^\bullet} & \s0  \\
\end{block}
\s{d_j}& \s0\phanbull &\s0\phanbull &\s1
\end{blockarray}
\qquad
\newJac = 
\begin{blockarray}{r@{\hspace{3mm}}ccc@{\hspace{10pt}}cc}
&  x_2\phanbull &  x_1 & \phanbull  x'_3 \phanbull  &  \\
\begin{block}{r @{\hspace{10pt}}[ccc]@{\hspace{2mm}}cc}
f_1 & 1  &\lbar1  & &     \\ \cline{2-3}
\newf_2 &  &\lbar{1}  & \lbar{g'_1(t)\phanbull} &    \\ \cline{3-4}
f_3 &  &  & \lbar1 &   \\
\end{block}
&
\end{blockarray}
\end{equation*}
The SA succeeds as $\newJ$ is nonsingular. The conversion results in a decrease in the value of the signature matrix: $\val\newSig=1<2=\val\Sig$.

\medskip

The basic ES method can work on \rf{blockLC1} by choosing $\vecv=[1,-1,0]^T\in\ker{\Jac}$. It is simpler---though trivial for this example---to work on block 1 only. We find $\vecv=[1,-1]^T\in\ker{\Jblk{1}{1}}$, and use \defterm{\ESdefEqnum} to derive
\begin{align*}
\nzsetES=\setbg{ \indxk\in B_1  \mid v_\indxk\nequiv 0} = \setbg{1,2},
\quad 
s=|\nzsetES|=2, 
\quad
\eqsetIES=\setbg{1,2}, 
\quad 
\overC=\max_{i\in\eqsetIES} c_i=1\;.
\end{align*}
Since $\vecv$ is constant, it is not difficult to verify that the ES conditions \defterm{\EScondEqnum} hold.

We choose $\indxk=2\in\nzsetES$ and introduce for $x_1$ a new variable 
\[
y_1=x_1^{(d_1-\overC)} - \frac{v_1}{v_2} \cdot x_2^{(d_2-\overC)} = x_1+x_2\;.
\]
The ES method hence says: replace $x_1$ by $y_1-x_2$ in $f_1$, and replace $x'_1$ by $y'_1-x'_2$ in $f_2$.
Finally we append the equation $g_1$ that prescribes such replacements, and obtain
\begin{equation*}
\begin{alignedat}{20}
0 &= \newf_1 &&= y_1 + h_1(t)                 &\qquad& 0 &&= \newf_3 &&= x'_3 +  h_3(t) \\
0 &= \newf_2 &&= x_1+y'_1 x'_3 + h_2(t) && 0 &&= g_1 &&= -y_1+x_1+x_2\;.
\end{alignedat}
\end{equation*}
\begin{equation*}
\newSig = 
\begin{blockarray}{r@{\hspace{3mm}}cccccc}
&  x_2\phanbull &   x_1 & y_1 & x_3  & \s{c_i} \\
\begin{block}{r @{\hspace{10pt}}[cccc]cc}
g_1 & 0^\bullet  &\lbar{0\phanbull}  & 0\phanbull  & &   \s0  \\ \cline{2-3}
\newf_2 &   &\lbar{0^\bullet}  & \lbar{1\phanbull} & 1\phanbull &  \s0  \\ \cline{3-4}
\newf_1 &   &  & \lbar{0^\bullet} & \dbar &   \s1  \\ \cline{4-5}
f_3 &  & &  & \lbar{1^\bullet} & \s0  \\
\end{block}
\s{d_j}& \s0\phanbull &\s0\phanbull &\s1\phanbull &\s1\phanbull
\end{blockarray}
\qquad
\newJac = 
\begin{blockarray}{r@{\hspace{3mm}}ccccc}
&  x_2\phanbull &   x_1 & y'_1 & x'_3  \\
\begin{block}{r @{\hspace{10pt}}[cccc]c}
g_1 & 1  &\lbar{1}  & -1  &\\ \cline{2-3}
\newf_2 &   &\lbar{1}  & \lbar{x'_3} & y'_1 \\ \cline{3-4}
\newf'_1 &   &  & \lbar{1} & \dbar  \\ \cline{4-5}
f_3 &  & &  & \lbar{1}  \\
\end{block}
&
\end{blockarray}
\end{equation*}
Again $\val\newSig=1<2=\val\Sig$, and the SA succeeds as $\det(\newJ)=1$.

\subsection{Block linear combination method.}\label{ssc:LCBTF}
 
We first introduce some convenient notation for the block LC method. Let $\~0_{\idr}$ denote the zero column vector of size $\idr$. Assume that a $\Jblk{\idq}{\idq}$ is identically singular. Let $\~\hu\in\coker{\Jblk{\idq}{\idq}}$, where $\~\hu\nequiv \~0_{N_\idq}$.  Let also 
\[
\vecu = \begin{bmatrix}
\~0_{N_1+\cdots+N_{\idq-1}} \\ 
\~\hu \\ 
\~0_{N_{\idq+1}+\cdots+N_p}
\end{bmatrix}\;.
\]

Denote 
\begin{equation}\label{eq:LCblkdef}
\begin{aligned}
\eqsetI &= \setbg{i \mid  u_i\nequiv 0}\subseteq B_\idq, \quad
\LCconst = \min_{i \in \eqsetI} c_{i}, 
\\ 
\nzset &= \setbg{\indxk\in \eqsetI \mid c_\indxk=\LCconst },
\qaq
\newL = \setbg{ l \in\nzset \mid \text{$u_l$ is (nonzero) constant}}\;.
\end{aligned}
\end{equation}
The set $\newL$ is used to seek a conversion that guarantees equivalence between the original DAE and the converted one. The block LC method is based on the following theorem.
\begin{theorem}\label{th:LCblk}
If 
\begin{equation}\label{eq:LCblkcond}
\hoder{x_j}{\vecu}< d_j-\LCconst \qquad \text{for all $j\in B_{\idq}$}
\end{equation}
and we replace an equation $f_\indxk$, $\indxk\in\nzset$, by
\begin{align*}
\newf = \sum_{i\in \eqsetI} u_i f_i^{\left(c_i-\LCconst\right)}\;,
\end{align*}
then $\val{\newSig} < \val{\Sig}$, where $\newSig=(\newsij{i}{j})$ is the signature matrix of the resulting DAE.
\end{theorem}

Before proving this theorem, we show how to apply the block LC method and prove a related lemma.
\begin{example}\label{ex:robotarmLC}
We illustrate the block LC method with the Campbell-Griepentrog two-link robot arm DAE \cite{campbell1995solvability}. We slightly simplify the problem formulation to \rf{robotarm}, allowing the first-order derivatives $x'_1$, $x'_2$, and $x'_3$ to occur implicitly in the equations. 
The two state variables $u_1$ and $u_2$ in the original formulation are renamed $x_4$ and $x_5$, respectively (and not to be confused with entries in a vector $\vecu$ in our notation).

The equations of this problem are
\begin{equation}\label{eq:robotarm}
\begin{aligned}
0 = A =\ & x''_1 - \Bigl[2c(x_3)(x'_1+x'_3)^2 + x'^2_1 d(x_3) + (2x_3-x_2)\bigl(a(x_3)+2b(x_3)\bigr) \\
&\phantom{x''_1} + a(x_3)(x_4-x_5)\Bigr] \\
0 = B =\ & x''_2 -\Bigl[ -2c(x_3)(x'_1+x'_3)^2-x'^2_1 d(x_3) +(2x_3-x_2)\bigl(1-3a(x_3)-2b(x_3) \bigr)  \\
 &\phantom{x''_2} - a(x_3)x_4+ \bigl( a(x_3)+1\bigr) x_5 \Bigr]\\
0 = C =\ & x''_3 -\Bigl[- 2c(x_3)(x'_1+x'_3)^2-x'^2_1 d(x_3)
 +(2x_3-x_2)\bigl( a(x_3)-9b(x_3)\bigr)\\ 
&\phantom{x''_3} -2x'^2_1 c(x_3) - d(x_3)\bigl( x'_1+x'_3\bigr)^2 -\bigl( a(x_3)+b(x_3) \bigr) (x_4-x_5)\Bigr] \\
0 = D =\ & \cos x_1 + \cos(x_1+x_3) - p_1(t)\\
0 = E =\ & \sin x_1 + \sin(x_1+x_3) - p_2(t)\;,
\end{aligned}
\end{equation}
where
\begin{align*}
\begin{array}{ll}
\begin{aligned}
a(\theta) &= 2/(2-\cos^2 \theta) \\
c(\theta) &=\sin \theta/(2-\cos^2 \theta) \\
p_1(t)&= \cos(1-e^t)+\cos(1-t) \\
\end{aligned}
&\qquad
\begin{aligned}
b(\theta) &= \cos \theta/(2-\cos^2 \theta) \\
d(\theta) &= \sin \theta\cos \theta/(2-\cos^2 \theta) \\
p_2(t) &= \sin(1-e^t)+\sin(1-t) \;.
\end{aligned}
\end{array}
\end{align*}

\begin{align*}
{\small
\Sig =
\begin{blockarray}{r@{\hspace{1mm}}r*5{@{\hspace{1mm}}c@{\hspace{1mm}}}  @{\hspace{1mm}} c}
 && x_{2} & \phanbull x_{4} & \phanbull x_{5} & \phanbull x_{1} & \phanbull x_{3} & \s{c_i} \\
\begin{block}{r@{\hspace{1mm}}r @{\hspace{2mm}}[*5{@{\hspace{1mm}}c@{\hspace{1mm}}}]@{\hspace{1mm}}c}
f_{1}&B&\phanbull2^\bullet&\lbar{0}&0&\phanbull\;\,\OK{1}&\OK{1}&\s{0}\\\cline{3-5}
f_{2}&C&\OK{0}&\lbar{\phanbull0^\bullet}&0&\lbar{\OK{1}}&2&\s{0}\\
f_{3}&A&\OK{0}&\lbar{0}&\phanbull0^\bullet&\lbar{2}&\OK{1}&\s{0}\\\cline{4-7}
f_{4}&D &&&&\lbar{\phanbull0^\bullet}&0&\s{2}\\
f_{5}&E &&&&\lbar{0}&\phanbull0^\bullet&\s{2}\\
\end{block}
&\s{d_j} &\s{2}&\phanbull\s{0}&\s{0}&\phanbull\s{2}&\s{2} \\
 \end{blockarray}\quad
\Jac = 
\begin{blockarray}{lccccccc}
 &x''_{2} &x_{4} &x_{5} &x''_{1} &x''_{3} \\
\begin{block}{l @{\hspace{10pt}}[ccccc]cc}
B &1&\lbar{a_3}&-a_3-1&&\\\cline{2-4}
C &&\lbar{\ a_3+b_3}&-a_3-b_3&\lbar{}&1\\
A & &\lbar{-a_3}&a_3&\lbar{1}&\\\cline{3-6}
D'' &&&&\lbar{\ \pp{D}{x_1}}&\pp{D}{x_3}\\
E'' &&&&\lbar{\ \pp{E}{x_1}}&\pp{E}{x_3}\\
\end{block}
&
\end{blockarray}
}
\end{align*}
Here in $\Jac$,
\begin{align*}
\begin{alignedat}{5}
a_3&=a(x_3)=2/(2-\cos^2 x_3) && &b_3&=b(x_3) &&=\cos x_3/(2-\cos^2 x_3)\\
\ppin{D}{x_1} &= - \sin x_{1} - \sin(x_{1} + x_{3})     &\qquad\qquad&&&\ppin{D}{x_3} &&= - \sin(x_{1} + x_{3}) 
\\ \ppin{E}{x_1} &= \phantom{-}\cos x_{1} + \cos(x_{1} + x_{3})
&&&& \ppin{E}{x_3} &&= \phantom{-}\cos(x_{1} + x_{3})\;.
\end{alignedat}
\end{align*}

The DAE \rf{robotarm} is of differentiation index 5, while the SA reports structural index $\nu_S=3$. Hence this must be a failure case, because $\nu_S$ is an upper bound for the differentiation index when the SA succeeds \cite{pryce2001a}.
We can see that the sub-Jacobian $\Jblk{2}{2}$ of block 2 is identically singular.

Our method first computes $\~\hu = [2, 2+\cos x_3]^T \in\coker{\Jblk{2}{2}}$. Then $\vecu=[0, 2, 2+\cos x_3, 0, 0]^T$. Using \rf{LCblkdef}, we have
\begin{align*}
\eqsetI &= \setbg{ i \mid  u_i\nequiv 0}= \setbg{2,3}, 
\quad \LCconst = \min_{i \in \eqsetI} c_{i} = 0, 
\quad \nzset =  \set{2,3}, 
\qaq \newL = \set{2}\;.
\end{align*}
The variables $x_4$ and $x_5$ in block 2 do not occur in $\vecu$, so the condition \rf{LCblkcond}  is satisfied. 

Considering equivalence, we pick $\indxk=2\in\newL$ over $\indxk=3\in\nzset\setminus\newL$, and replace $f_\indxk=C$ by $\overline C = u_1C + u_2 A = 2C + (2+\cos x_3) A$.
The SA results of the resulting DAE are as follows.
\begin{equation*}
{\small
\begin{alignedat}{3}
\newSig &=
\begin{blockarray}{r   *5{@{\hspace{1mm}}c@{\hspace{1mm}}}  @{\hspace{1mm}} c}
 & x_{4} & x_{5} & \phanbull x_{2} & x_{1} & x_{3} & \s{c_i}  \\
\begin{block}{r @{\hspace{2mm}}[*5{@{\hspace{1mm}}c@{\hspace{1mm}}}]@{\hspace{1mm}}c}
A&0^\bullet&0\phanbull&\lbar{\OK{0}}&\OK{2}&\OK{1}&\s{0}\\
B&0\phanbull&0^\bullet&\lbar{2}&\OK{1}&\OK{1}&\s{0}\\\cline{2-4}
 \overline C&&&\lbar{\phanbull 0^\bullet}&\lbar{2\phanbull }&2 &\s{2}\\\cline{4-6}
D&&&&\lbar{0^\bullet}&0&\s{4}\\
E&&&&\lbar{0\phanbull}&\phanbull 0^\bullet&\s{4}\\
\end{block}
\s{d_j} &\s{0}&\s{0}\phanbull &\phanbull\s{2}&\s{4}&\s{4}  &  \\
\end{blockarray} \quad
\newJ = 
\begin{blockarray}{l*5{c@{\hspace{2mm}}}}
 &x_{4} &x_{5} &x''_{2} &x^{(4)}_{1} & x^{(4)}_{3} \\
\begin{block}{l @{\hspace{10pt}}[ccccc]}
A&-a_3&a_3&\lbar{}&&\\
B&\phantom{-}a_3 &-a_3-1  &\lbar{1} &&\\\cline{2-4}
&&&\dbar&\dbar\\[-2ex]
\overline C''&&&\lbar{\ \ \pp{\overline C}{x_2}\ \ }&\lbar{\ \ 2+\cos x_3\ \ }& 2\\ \cline{4-6}
D^{(4)}&&&&\lbar{\ \pp{D}{x_1}}&\pp{D}{x_3}\\
E^{(4)}&&&&\lbar{\ \pp{E}{x_1}}&\pp{E}{x_3}\\
\end{block}
\end{blockarray}
\end{alignedat}
}
\end{equation*}
Here $\ppin{\overline{C}}{x_2} = 2(a_3^2-3a_3b_3+b_3^2)(2-\cos^2 x_3)$. The SA reports $\nu_S=5$ and succeeds at any point where \[\det(\newJ) = 4(a_3^2-3a_3b_3+b_3^2)\sin x_3 \neq 0\;.\] Now $\val\newSig = 0<2 = \val\Sig$.

\end{example}

\begin{lemma}\label{le:LCsigma}
Consider a BTF of a Jacobian pattern $\SpJac$ derived from $\Sig$ and $\offsetpair$. 
If we perform the LC conversion as described in \THref{LCblk}, then in the resulting $\newSig$,
\begin{align}\label{eq:LCsigma}
d_j - c_i
\left\{
\begin{alignedat}{7}
&>\newsij{i}{j}     &\quad& \text{if $\blk{j} < \blk{i}$}\\
&\ge \newsij{i}{j}  && \text{if $\blk{j} \ge \blk{i}$}\;.
\end{alignedat}
\right.
\end{align}
\end{lemma}
\begin{proof}
We only replace $f_\indxk$ by $\newf_\indxk=\newf$ in a conversion, so $\newsig_{ij} = \sij{i}{j}$ for all $i\neq\indxk$ and all $j$. By \rf{djciblk}, \rf{LCsigma} holds for all $i\neq\indxk$.

When $i=\indxk$, we consider two cases:  (a) $\blk{j}<\idq$ and (b) $\blk{j}\ge\idq$.

\medskip

\noindent (a)\ \ $\blk{j}<\idq=\blk{\indxk}$. By \rf{djciblk}, $\sij{\indxk}{j}< d_j-c_\indxk$. Then $\newsij{\indxk}{j}$ is 
\begin{align}\label{eq:LCblkbelow}
\hoder{x_j}{\newf_\indxk}  
= \sigma\Bigl({x_j},{\sum_{i\in \eqsetI} u_i f_i^{\left(c_i-\LCconst\right)}}\Bigr) 
\le \max\Bigl\{ \hoder{x_j}{\vecu},\, \max_{i\in\eqsetI} \sigma\bigl({x_j},{f_i^{(c_i-\LCconst)}}\bigr) \Bigr\} \;. 
\end{align}
We use some simple derivations to obtain
\begin{subequations}\label{eq:LCblkbelow12}
\begin{align}
\hoder{x_j}{\vecu} 
&\le \hoder{x_j}{\Jblk{\idq}{\idq}}  
\le \max_{i\in\eqsetI} \hoder{x_j}{f_i}= \max_{i\in\eqsetI} \sij{i}{j}  \nonumber \\
&< \max_{i\in\eqsetI} (d_j - c_i)
= d_j - \min_{i\in\eqsetI} c_i 
= d_j - c_l \label{eq:LCblkbelow1} \quad\text{and} \\
\max_{i\in\eqsetI}\hoder{x_j}{f_i^{(c_i-\LCconst)}}
&= \max_{i\in\eqsetI} \left(   \sij{i}{j}+c_i-\LCconst \right)  
<  d_j - \LCconst
=  d_j-c_\indxk\;. \label{eq:LCblkbelow2}
\end{align}
\end{subequations}
Substituting \rf{LCblkbelow1} and \rf{LCblkbelow2} in  \rf{LCblkbelow}, we obtain $\newsij{\indxk}{j} =\hoder{x_j}{\newf_\indxk}< d_j - c_\indxk$.

\medskip

\noindent (b)\ \  $\blk{j}\ge\idq=\blk{\indxk}$. By \rf{djciblk}, $\sij{\indxk}{j}\le d_j-c_\indxk$. We can replace the two ``$<$'' in \rf{LCblkbelow12}  by ``$\le$'', and using these inequalities in \rf{LCblkbelow}, we have $\newsij{\indxk}{j} \le d_j - c_\indxk$. \qed

\end{proof}

Using \LEref{LCsigma}, we can now prove \THref{LCblk}.

\begin{proof}
By  \LEref{BTF4.2}, we  can regard block $\idq$ as a sub-DAE with $\Sig_{\idq\idq}$ and $(\~c_\idq; \~d_\idq)$. The conversion described in \THref{LCblk} can be considered as an application of the basic LC method to this sub-DAE. Since the block LC condition \rf{LCblkcond} holds, that is, $\hoder{x_j}{\vecu}< d_j-\LCconst$ for all $j\in B_\idq$ that belong to this sub-DAE, the basic LC condition \defterm{\LCcondEqnum} also holds for the sub-DAE. Hence $\val{\newSig_{\idq\idq}}<\val{\Sig_{\idq\idq}}$.

Let $\newT$ be a HVT of $\newSig$. Using \rf{LCsigma} in \LEref{LCsigma}, we have 
$
\val{\newSig}=\sum_{(i,j)\in\newT} \newsij{i}{j} \le d_j-c_i = \val\Sig
$.
Now we prove $\val{\newSig}<\val{\Sig}$ by contradiction. First assume that $\val{\newSig}=\sum_j d_j-\sum_i c_i=\val{\Sig}\ge0$. With \rf{LCsigma} , the three conditions in \LEref{blockstructure} are satisfied.  From this lemma, it follows that the Jacobian patterns $\overline{\~S}_0$, derived from $\newSig$ and $\offsetpair$, and $\SpJac$, derived from $\Sig$ and $\offsetpair$, are in the same  $p\times p$ BTF.

By \LEref{unionHVT},  $\newT$ is the union of HVTs $\newT_\idv$ of all diagonal blocks $\newSig_{\idv\idv}$, $\idv=\rnge{1}{p}$. 
By the construction of $\newSig$, $\val{\newSig_{\idv\idv}}=\val{\Sig_{\idv\idv}}$ for all $\idv\neq\idq$. Then a contradiction follows from
\begin{align}\label{eq:sumdiag}
\nonumber\val{\newSig}  
&= \sum_{(i,j)\in\newT} \newsij{i}{j}
= \sum_{\idv=1}^p  \sum_{(i,j)\in\newT_\idv} \newsij{i}{j}
= \sum_{\idv=1}^p \val{\newSig_{\idv\idv}}  
\\\nonumber&= \sum_{\idv\neq\idq}\val{\newSig_{\idv\idv}} + \val{\newSig_{\idq\idq}}
< \sum_{\idv\neq\idq}\val{\Sig_{\idv\idv}} + \val{\Sig_{\idq\idq}}
\\&= \sum_{\idv=1}^p \val{\Sig_{\idv\idv}} 
= \sum_{\idv=1}^{p}  \sum_{(i,j)\in T_\idv} \sij{i}{j}
= \sum_{(i,j)\in T} \sij{i}{j}
= \val{\Sig}\;,
\end{align}
where $T$ is a HVT of $\Sig$ and $T_\idv$ are HVTs of its diagonal blocks $\Sig_{\idv\idv}$. 
The assumption $\val\newSig=\val\Sig$ is hence false, so $\val\newSig<\val\Sig$ holds.
\end{proof}

\subsection{Block expression substitution method.}\label{ssc:ESBTF}

Assume again that a $\Jblk{\idq}{\idq}$ is identically singular. Let  $\~\hv\in\ker{\Jblk{\idq}{\idq}}$, where $\~\hv\nequiv \~0_{N_\idq}$. Similarly, we construct the column $n$-vector 
\[
\vecv = \begin{bmatrix}
\~0_{N_1+\cdots+N_{\idq-1}} \\ 
\~\hv \\ 
\~0_{N_{\idq+1}+\cdots+N_p}
\end{bmatrix}\;.
\]

We use notation similar to that used in the basic ES method (see \defterm{\S 4.2})\,:
\begin{equation}\label{eq:ESblkdef}
\begin{aligned}
\nzsetES &=\setbg{ j\mid v_j\nequiv 0}\subseteq B_\idq, 
\quad  \eqsetIES = \setbg{ i\in B_\idq \mid \text{$d_j-c_i=\sij{i}{j}$ for some $j\in \nzsetES$} }\;,  \\
s &= |\nzsetES|, 
\quad \ESconst = \max_{i\in \eqsetIES} c_i
\qaq \newJes = \setbg{ \indxk \mid \text{$v_\indxk$ is (nonzero) constant}} \;.
\end{aligned}
\end{equation}
The set $\newJes$ is used to seek a conversion that guarantees equivalence between the original DAE and the converted one.
The conditions for applying the block ES method~are
\begin{equation}\label{eq:ESblkcond}
\begin{aligned}
\hoder{x_j}{\vecv}  &\casemod{ll}{
< d_j-\ESconst     & \quad\text{if $j\in \nzsetES\ $ or $\ \blk{j}<\idq$ 
} \\[1ex]
\le d_j-\ESconst    & \quad\text{if $j\in B_\idq \setminus \nzsetES\ $ or $\ \blk{j}>\idq$,
}
} \quad\text{and}\\
d_j-\ESconst &\ge 0 \quad \text{for all $j\in\nzsetES$\;.} 
\end{aligned}
\end{equation}

We choose an $l\in\nzsetES$, and introduce $s-1$ new variables
\begin{align}\label{eq:ESblkgjgk}
y_j &= 
x_j^{(d_j-\ESconst)} - \frac{v_j}{v_\indxk} \cdot x_\indxk^{(d_\indxk-\ESconst)} \quad\text{for all $j\in\Jnol$}\;.
\end{align}
In each $f_i$ with $i\in B_\idq$, we
\begin{equation}\label{eq:ESblksubs}
\begin{aligned}
\text{replace each} &\quad \text{$x_j^{(\sij{i}{j})}\quad$ with $d_j-c_i=\sij{i}{j}$ and $j\in\Jnol$} \\
\text{by} &\quad \Big( y_j + \frac{v_j}{v_\indxk} \cdot x_\indxk^{(d_\indxk-\ESconst)} \Big)^{(\ESconst-c_i)}
\;.
\end{aligned}
\end{equation}
Note that because of $\eqsetIES$ in \rf{ESblkdef}, we actually perform replacements (equivalently referred to as ``expression substitutions'') in only $f_i$'s with $i\in \eqsetIES\subseteq B_\idq$. Denote each new $f_i$ by $\newf_i$, and let also $\newf_i=f_i$ for the unchanged equations with $i\notin \eqsetIES$. 

By \rf{ESblkgjgk}, we append $s-1$ equations  that prescribe the substitutions in \rf{ESblksubs}:
\begin{align}\label{eq:ESblkgj}
0 = g_j &= -y_j + x_j^{(d_j-\ESconst)} - \frac{v_j}{v_\indxk} \cdot x_\indxk^{(d_\indxk-\ESconst)} \quad \text{for all $j\in\Jnol$}\;.
\end{align}
The block ES method is based on the following theorem.

\begin{theorem}\label{th:ESblk}
Let $\nzsetES$, $s$, $\eqsetIES$, and $\ESconst$ be as defined in \rf{ESblkdef}. Assume that the conditions \rf{ESblkcond} hold. For an $\indxk\in\nzsetES$, if we 
\begin{enumerate}[1)]
\item introduce $s-1$ new variables $x_j$, $j\in\Jnol$, as defined in \rf{ESblkgjgk},
\item perform replacements in $f_i$, for all $i\in B_\idq$, as described in \rf{ESblksubs}, and 
\item append $s-1$ equations $g_j$, $j\in\Jnol$, as defined in \rf{ESblkgj},  
\end{enumerate}
then $\val{\newSig} < \val{\Sig}$, where $\newSig$ is the signature matrix of the resulting DAE.
\end{theorem}

Before proving this theorem, we illustrate the block ES method with the robot arm DAE \rf{robotarm} and prove two related lemmas.

\begin{example}\label{ex:robotarmES}
The method finds first $\~\hv=[1,1]^T\in\ker{\Jblk{2}{2}}$. Then $\vecv=[0,0,1,1,0]^T$. Using \rf{ESblkdef}, we have
\begin{align*}
\nzsetES = \newJes = \setbg{ j\mid v_j\nequiv 0} = \set{2,3},
\quad s = |\nzsetES|=2,
\quad \eqsetIES = \set{2,3}, \quad \ESconst = \max_{i\in \eqsetIES} c_i = 0\;.
\end{align*}
Since $\vecv$ is constant, $\nzsetES = \newJes$ and the first condition in \rf{ESblkcond} holds. The second condition in \rf{ESblkcond} holds also, as $d_4-\ESconst=d_5-\ESconst=0$.  We choose $x_4$, whose column index in the permuted $\Sig$ is $\indxk = 2\in\newJes$. Then we introduce for $x_5$, the other variable in block 2 with column index $j=3$, a new variable
\begin{align*}
y_5 = x_5^{(d_5-\ESconst)} - \frac{v_3}{v_2} \cdot x_4^{(d_4-\ESconst)} = x_5-x_4\;.
\end{align*}
Correspondingly, we append $0=g_5=-y_5+ x_5-x_4$ and  replace $x_5$ by $y_5+x_4$ in $C$ and $A$, the equations in block 2. 

The resulting DAE has the following new equations
\begin{equation*}\label{eq:robotarmES}
\begin{alignedat}{5}
0 &= \overline{A} &&=  x''_1 - \Bigl[2c(x_3)(x'_1+x'_3)^2 + x'^2_1 d(x_3) + (2x_3-x_2)\bigl(a(x_3)+2b(x_3)\bigr) + a(x_3)y_5 \Bigr] \\
%
%
0 &= \overline{C}  &&= x''_3 -\Bigl[- 2c(x_3)(x'_1+x'_3)^2-x'^2_1 d(x_3) + (2x_3-x_2)\bigl( a(x_3)-9b(x_3)\bigr)\\ 
&&&\phantom{=x''_3-\Big[\,} -2x'^2_1 c(x_3) - d(x_3)\bigl( x'_1+x'_3\bigr)^2 -\bigl( a(x_3)+b(x_3) \bigr) y_5 \Bigr]\\
%
0 &= g_5 &&= -y_5 + x_4-x_5\;.
\end{alignedat}
\end{equation*}

{\footnotesize
\begin{align*}
\newSig =
\begin{blockarray}{r*6{@{\hspace{0.5mm}}c@{\hspace{0.5mm}}}  @{\hspace{0.5mm}} c}
 & x_{4} & x_{5} & x_{2} & y_5 & x_{1} & x_{3} & \s{c_i}  \\
\begin{block}{r @{\hspace{2mm}}[*6{@{\hspace{0.5mm}}c@{\hspace{0.5mm}}}]@{\hspace{0.5mm}}c}
g_5 &\phanbull0^\bullet&0&\dbar&\OK{0}&&&\s{0}\\
B&0&\phanbull0^\bullet&\lbar{2}&&\phanbull\OK{1}&\OK{1}&\s{0}\\\cline{2-5}
\overline C&&&\lbar{\phanbull0^\bullet}&0&\lbar{\OK{1}}&2&\s{2}\\
\overline A&&&\lbar{0}&\phanbull0^\bullet&\lbar{2}&\OK{1}&\s{2}\\\cline{4-7}
D&&&&&\lbar{\phanbull0^\bullet}&0&\s{4}\\
E&&&&&\lbar{0}&\phanbull0^\bullet&\s{4}\\
\end{block}
\s{d_j} &\s{0}&\s{0}&\s{2}&\s{2}&\s{4}&\s{4}&\\
\end{blockarray} 
\quad
\newJ = 
\begin{blockarray}{l *6{@{\hspace{0.5mm}}c@{\hspace{0.5mm}}}  }
 &x_{4} &x_{5} & x''_{2} & y''_5 & x^{(4)}_{1} & x^{(4)}_{3} \\
\begin{block}{l @{\hspace{2mm}}[*6{@{\hspace{0.5mm}}c@{\hspace{0.5mm}}}]}
g_5 &-1&1&\lbar{}&&&\\
B &a_3 &-a_3-1 &\lbar{1}&&&\\\cline{2-5}
\overline C''&&&\lbar{\ a_3-9b_3\ }&\ \ -a_3-b_3\ \ &\lbar{}&1\\
\overline A''&&&\lbar{\ a_3+2b_3\ }& a_3 &\lbar{1}&\\\cline{4-7}
D^{(4)}&&&&&\lbar{\ \ \pp{D}{x_1}\ \ } & \pp{D}{x_3}\\
E^{(4)}&&&&&\lbar{\ \ \pp{E}{x_1}\ \ } & \pp{E}{x_3}\\
\end{block}
\end{blockarray}
\end{align*}
}Now the System Jacobian $\newJ$ is generically nonsingular. The SA reports the correct index 5, and succeeds at any point where $\det(\newJ) = 2(a_3^2-3a_3b_3+b_3^2)\sin x_3 \neq 0$. Again $\val\newSig = 0<2 = \val\Sig$.

\end{example}

In \cite{Pryce98}, Pryce fixed the SA's failure on \rf{robotarm}, and pointed out that {\em only} the introduction of $x_4-x_5$ as a separate variable is essential to his fix. \EXref{robotarmES} verifies Pryce's argument and shows that the block ES method finds his reformulation in a systematic way.

\medskip

To prove \THref{ESblk}, we shall use the following two assumptions.
\begin{enumerate}[(a)]
\item Without loss of generality, we assume the entries $\hv_j\nequiv 0$ are in the first $s$ positions of $\~\hv$, that is, $\~\hv = [\hv_1,\ldots, \hv_s,0,\ldots ,0]^T$. By \rf{ESblkdef}, 
$
\nzsetES = \rnge{\sum_{\idv=1}^{\idq-1} N_\idv+1}{\sum_{\idv=1}^{\idq-1} N_\idv+s}
$.
\item We introduce one more variable $y_\indxk=x_\indxk^{(d_\indxk-\ESconst)}$ for the chosen $\indxk\in\nzsetES$, and append correspondingly one more equation $0=g_\indxk = -y_\indxk + x_\indxk^{(d_\indxk-\ESconst)}$.
\end{enumerate}

We show first that the signature matrix $\newSig$ of the resulting DAE can be put in the block structure as shown in \FGref{ESblkSig}. Then we construct two $(n+s)$-vectors $\~{\newc}$ and $\~{\newd}$ in \rf{ESblknewcd}, and prove in \LEref{ESblklemma} that $\newd_j-\newc_i>\newsij{i}{j}$ holds below the block diagonal, while $\newd_j-\newc_i\ge \newsij{i}{j}$ holds elsewhere. 
Lastly, we prove \THref{ESblk}.

\begin{figure}[!th]
\centering
\includegraphics[trim = 20mm 205mm 72mm 38mm, clip, scale=1]{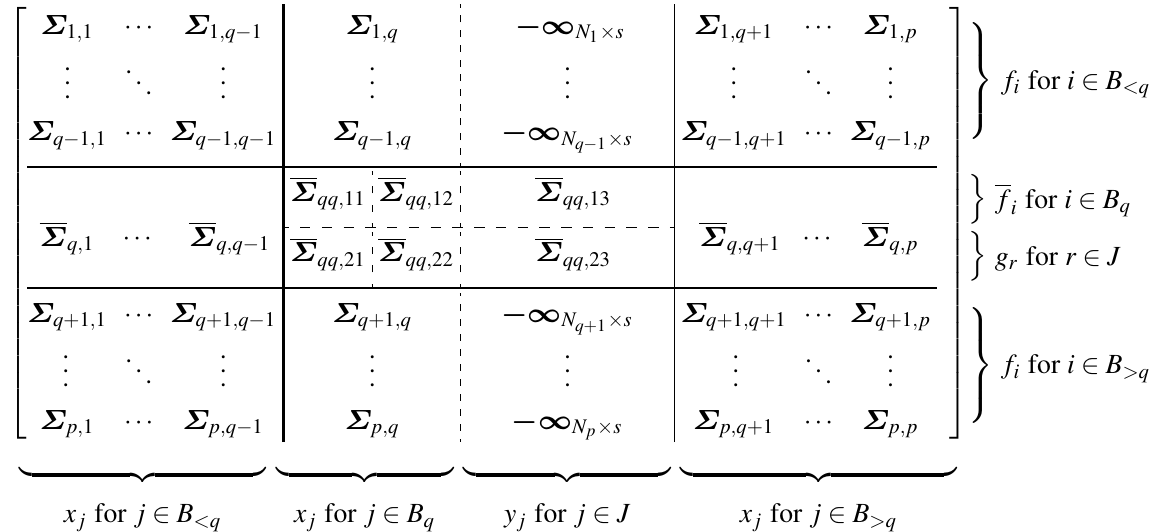}
\caption{\label{fg:ESblkSig} 
Block structure of $\newSig$ of the resulting DAE by the block ES method. 
The notation $B_{<\idq}$  is short for $\cup_{\idv=1}^{\idq-1} B_w$, and $B_{>\idq}$  is short for $\cup_{\idv=\idq+1}^{p} B_\idv$. 
}
\end{figure}

From the description of the conversion in \THref{ESblk}, the substitutions \rf{ESblksubs} only occur in equations $f_i$ with $i\in B_\idq$. Hence, in the resulting DAE, variables $y_j$ for $j\in\nzsetES$ only appear in $\newf_i$ for $i\in B_\idq$ and $g_r$ for $r\in\nzsetES$.

Considering the block structure of $\newSig$ in \FGref{ESblkSig}, we distinguish between the four cases for a block submatrix $\newSig_{\idv_1 \idv_2 }$: 
(a) $\idv_1\neq\idq$ and $\idv_2\neq \idq$, 
(b) $\idv_1\neq\idq$ and $\idv_2=\idq$, 
(c) $\idv_1=\idq$ and $\idv_2\neq\idq$, and 
(d) $\idv_1=\idv_2=\idq$.

\begin{enumerate}[(a)]
\item $\idv_1\neq\idq$ and $\idv_2\neq \idq$. In $\newSig_{\idv_1 \idv_2}$, equations $f_i$ are of indices $i\in B_{<\idq}\cup B_{>\idq}$. As noted in \rf{ESblkdef}, the expression substitutions described in \rf{ESblksubs} only take place in $f_{i'}$ with $i'\in \eqsetIES \subseteq B_\idq$, so do not happen in such blocks $\newSig_{\idv_1 \idv_2}$. Hence, each $\Sig_{\idv_1\idv_2}$ remains unchanged in $\newSig$: $\newSig_{\idv_1 \idv_2 } = \Sig_{\idv_1 \idv_2 }$ for $\idv_1\neq\idq$ and $\idv_2\neq \idq$.

\item $\idv_1\neq \idq$ and $\idv_2=\idq$. 
In $\newSig_{\idv_1 \idq}$, we include variables $y_j$ for $j\in\nzsetES$ as defined in \rf{ESblkgjgk}. 
By the same arguments as in (a), the expression substitutions do not happen in these blocks. That is, $y_j$ for $j\in\nzsetES$ do not appear in equations $f_i$ for $i\in B_{<\idq}\cup B_{>\idq}$. 
Hence, we can obtain $\newSig_{\idv_1\idq}$ by concatenating horizontally  $\Sig_{\idv_1\idq}$ with an $N_{\idv_1}\times s$ matrix of $-\infty$'s:
$\newSig_{\idv_1 \idq}=
\bigl[ \Sig_{\idv_1\idq},\ \bninfty_{N_{\idv_1}\times s}\bigr]$ for $m_1=1,\ldots,\idq-1,\idq+1,\ldots,p$.

\item $\idv_1=\idq$ and $\idv_2\neq\idq$. 
In $\newSig_{\idq \idv_2}$, we include equations $g_r$ for $r\in\nzsetES$ as defined in \rf{ESblkgj}. Also, due to the expression substitutions \rf{ESblksubs} occurring in $f_i$ with $i\in\eqsetIES\subseteq B_\idq$, $\hoder{x_j}{f_i}$ and $\hoder{x_j}{\newf_i}$ may not be the same for all $i\in B_\idq$ and all $j=\rnge{1}{n}$. Hence,~in contrast to cases (a) and (b),  there are no obvious connections between $\Sig_{\idv_1 \idv_2}$ and $\newSig_{\idv_1 \idv_2}$ for $\idv_1=\idq$ and $\idv_2\neq\idq$.

\item $\idv_1=\idv_2=\idq$. $\newSig_{\idq \idq}$ contains signature entries for equations $\newf_i$ and $g_r$, where $i\in B_\idq$ and $r\in\nzsetES$, in variables $x_j$ and $y_r$, where $j\in B_\idq$ and $r\in\nzsetES$.
Similar to the resulting signature matrix $\newSig$ by the basic ES method \cite[\S4.2]{tgn2015c}, $\newSig_{\idq \idq}$ by the block ES method also has a (sub)block structure; c.f \rf{ESblkqq}. We shall use it in the proof of \LEref{ESblklemma} in Appendix~\scref{proofESblknewSig}.
\end{enumerate}

Let $\idQ=\sum_{\idv=1}^\idq N_\idv$ denote the total number of equations (or variables) in the first $\idq$ blocks. Using a valid  $\offsetpair$ of $\Sig$, we construct an offset pair $(\~{\newc};\~{\newd})$ defined as
\begin{equation}\label{eq:ESblknewcd}
\newc_i =\casemod{ll}
{
c_i         &\text{if $i=\rnge{1}{\idQ}$}\\
\ESconst   &\text{if $i=\rnge{\idQ+1}{\idQ+s}$}\\
c_{i-s}   &\text{if $i=\rnge{\idQ+s+1}{n+s}$},
} 
\quad
\newd_j = \casemod{ll}
{
d_j         &\text{if $j=\rnge{1}{\idQ}$}\\
\ESconst   &\text{if $j=\rnge{\idQ+1}{\idQ+s}$}\\
d_{j-s}   &\text{if $j=\rnge{\idQ+s+1}{n+s}$}\;.
}
\end{equation}
Then we have the following lemma. Since its proof is rather technical, we present it in Appendix~\scref{proofESblknewSig}.

\begin{lemma}\label{le:ESblklemma}
Let $(\~{\newc};\~{\newd})$ be as constructed in \rf{ESblknewcd}. In the block structure of $\newSig$ in \FGref{ESblkSig}, 
if a position $(i,j)$ in $\newSig$ is below the block diagonal, then $\newd_j - \newc_i>\newsij{i}{j}$; otherwise, $\newd_j - \newc_i\ge\newsij{i}{j}$.
\end{lemma}

Using this lemma, we can now prove \THref{ESblk}.
\begin{proof}
Let $\newT$ be a transversal of $\newSig$. Using \LEref{ESblklemma} and  \rf{ESblknewcd}, we derive
\begin{alignat*}{3}
\val{\newSig} &
= \sum_{(i,j)\in \newT} \newsij{i}{j} 
\le \sum_{(i,j)\in \newT}(\newd_j-\newc_i) 
=\sum_{j=1}^{n+s} \newd_j - \sum_{i=1}^{n+s} \newc_i
\\&= \left(\sum_{j=1}^{Q} d_j + s\ESconst + \sum_{j=Q+s+1}^{n+s} d_{j-s}\right)
- \left(\sum_{j=1}^{Q} c_i + s\ESconst + \sum_{i=Q+s+1}^{n+s} c_{i-s}\right)
\\&= \sum_{j=1}^{n} d_j - \sum_{i=1}^{n} c_i 
= \val{\Sig}\;.
\end{alignat*}

As in the proof of \THref{LCblk}, we prove $\val{\newSig}<\val{\Sig}$ by contradiction. Assume $\val{\newSig}=\sum_j^{n+s} \newd_j - \sum_i^{n+s} \newc_i = \val{\Sig}$, which is $\ge0$. With \LEref{ESblklemma}, the three conditions in \LEref{blockstructure} are satisfied. Then it follows from \LEref{blockstructure} that
\begin{enumerate}[(a)]
\item $(\~\newc;\~\newd)$ is a valid offset pair of $\newSig$;
\item the Jacobian pattern $\overline{\~S}_0$, derived from $\newSig$ and $(\~\newc;\~\newd)$, is in the $p\times p$ BTF shown in \FGref{ESblkSig};
\item $\newT$ is the union of HVTs $\newT_\idv$ of all diagonal blocks $\newSig_{11},\ldots,\newSig_{pp}$ of $\newSig$. 
\end{enumerate}

We can consider block $\idq$ of the original DAE as a sub-DAE, with signature matrix $\Sig_{\idq\idq}$ and offset pair $(\~c_\idq; \~d_\idq)$---this follows from \LEref{BTF4.2}. The conversion described in \THref{ESblk} can be regarded as an application of the basic ES method to this sub-DAE, given that the ES conditions \defterm{\EScondEqnum} hold because of \rf{ESblkcond}. By \cite[Theorem 4.2]{tgn2015c} for the basic ES method, a conversion results in $\val{\newSig_{\idq\idq}}<\val{\Sig_{\idq\idq}}$.  Also, since $\newSig_{\idv\idv}=\Sig_{\idv\idv}$ for $\idv\neq\idq$, $\val{\newSig_{\idv\idv}}=\val{\Sig_{\idv\idv}}$. 
Then a contradiction $\val\newSig<\val\Sig$ follows by the same derivations as in \rf{sumdiag}. 
The assumption $\val\newSig=\val\Sig$ is hence false, so $\val\newSig<\val\Sig$ holds. \qed
\end{proof}

\section{More examples.}\label{sc:exam}
We demonstrate how to apply the block conversion methods on two DAE problems originated from electrical circuit analysis \cite{TestSetIVP}. They are the transistor amplifier and the ring modulator. We describe them  in \SSCref{transamp} and \SSCref{ringmod}, respectively.

\subsection{Transistor amplifier.}\label{ssc:transamp}
The transistor amplifier DAE is
\begin{equation*}
\begin{alignedat}{7}
0&=f_1&&=&&C_1(x'_1-x'_2) + R_0^{-1}(x_1-U_e(t)) \\
0&=f_2&&=-&&C_1(x'_1-x'_2) - R_2^{-1}U_b + x_2\left(R_1^{-1}+R_2^{-1}\right) -(\alpha-1)g(x_2-x_3)\\
0&=f_3&&= &&C_2x'_3 - g(x_2-x_3) + R_3^{-1}x_3\\
0&=f_4&&= &&C_3(x'_4-x'_5) + R_4^{-1}(x_4-U_b) + \alpha g(x_2-x_3)\\
0&=f_5&&= -&&C_3(x'_4-x'_5) - R_0^{-1}U_b + x_5\left(R_5^{-1}+R_6^{-1}\right) -(\alpha-1)g(x_5-x_6)\\
0&=f_6&&= &&C_4x'_6 - g(x_5-x_6) + R_7^{-1}x_6 \\
0&=f_7&&= &&C_5(x'_7-x'_8) + R_8^{-1}(x_7-U_b) + \alpha g(x_5-x_6)\\
0&=f_8&&= -&&C_5(x'_7-x'_8) + R_9^{-1}x_8\;.
\end{alignedat}
\end{equation*}
{\scriptsize
\begin{align*}
\Sig =
\begin{blockarray}{r@{\hspace{1mm}} *8{@{\hspace{0.5mm}} c @{\hspace{0.5mm}}} @{\hspace{0.5mm}}c}
 & x_{1} & x_{2} & x_{3} & x_{4} & x_{5} & x_{6} & x_{7} & x_{8} & \s{c_i} \\
\begin{block}{r @{\hspace{2mm}}[*8{@{\hspace{0.5mm}} c@{\hspace{0.5mm}}} @{\hspace{0.5mm}}] c}
 f_{1}&\phanbull1^\bullet&1&\dbar&&&&&&\s{0}\\
 f_{2}&1&\phanbull1^\bullet&\lbar{\OK{0}}&&&&&&\s{0}\\\cline{2-4}
 f_{3}&&\OK{0}&\lbar{\phanbull1^\bullet}&\dbar&&&&&\s{0}\\\cline{4-6}
 f_{4}&&\OK{0}&\phanbull\OK{0}&\lbar{\phanbull1^\bullet}&1&\dbar&&&\s{0}\\
 f_{5}&&&&\lbar{1}&\phanbull1^\bullet&\lbar{\OK{0}}&&&\s{0}\\\cline{5-7}
 f_{6}&&&&&\OK{0}&\lbar{\phanbull1^\bullet}&\dbar&&\s{0}\\\cline{7-9}
 f_{7}&&&&&\OK{0}&\phanbull\OK{0}&\lbar{\phanbull1^\bullet}&1&\s{0}\\
 f_{8}&&&&&&&\lbar{1}&\phanbull1^\bullet&\s{0}\\
\end{block}
\s{d_j} &\s{1}&\s{1}&\s{1}&\s{1}&\s{1}&\phanbull\s{1}&\phanbull\s{1}&\s{1} \\
\end{blockarray}
\qquad\Jac =
\begin{blockarray}{r@{\hspace{1mm}} *8{@{\hspace{0.5mm}} c @{\hspace{0.5mm}}} @{\hspace{0.5mm}}c}
 & x'_{1} & x'_{2} & x'_{3} & x'_{4} & x'_{5} & x'_{6} & x'_{7} & x'_{8} & \s{c_i} \\
\begin{block}{r @{\hspace{2mm}}[*8{@{\hspace{0.5mm}} c@{\hspace{0.5mm}}} @{\hspace{0.5mm}}] c}
 f_{1}&\phantom{-}C_1&-C_1&\dbar&&&&&&\s{0}\\
 f_{2}&-C_1&\phantom{-}C_1&\dbar&&&&&&\s{0}\\\cline{2-4}
 f_{3}&&&\lbar{C_2}&\dbar&&&&&\s{0}\\\cline{4-6}
 f_{4}&&&&\lbar{\phantom{-}C_3}&\ \ -C_3&\dbar&&&\s{0}\\
 f_{5}&&&&\lbar{-C_3}&\ \ \phantom{-}C_3&\dbar&&&\s{0}\\\cline{5-7}
 f_{6}&&&&&&\lbar{C_4}&\dbar&&\s{0}\\\cline{7-9}
 f_{7}&&&&&&&\lbar{\phantom{-}C_5}&\ \ -C_5&\s{0}\\
 f_{8}&&&&&&&\lbar{-C_5}&\ \ \phantom{-}C_5&\s{0}\\
\end{block}
\end{blockarray}
\end{align*}
}In the equations, $g(y) = \beta\left( e^{y/U_F}-1\right)$ and $U_e(t) = 0.1\sin (200\pi t)$; $\alpha, \beta, U_b, U_F, R_0$, $R_k$ for $k=\rnge{1}{9}$, and $C_k$ for $k=\rnge{1}{5}$ are positive constants. The SA fails since $\Jac$ is identically singular. The fine BTF reveals that the three $2\times 2$ sub-Jacobians $\Jblk{1}{1}$, $\Jblk{3}{3}$, $\Jblk{5}{5}$ are identically singular and have a similar structure. Each block receives the same treatment when a conversion method is applied.

\noindent{\em LC method.}
One can easily find $\~\hu=[1,1]^T\in \coker{\Jblk{1}{1}}, \coker{\Jblk{3}{3}}, \coker{\Jblk{5}{5}}$. We perform on each singular block a conversion, and choose to replace the first equation in each such block.
\[
  \begin{array}{l@{\hspace{3mm}}l@{\hspace{3mm}}l}
\text{block} & \text{replace} & \text{by}   \\ \hline \\[-2ex]
1 & f_1 & \newf_1 = f_1+f_2 \\
3 & f_4 & \newf_4 = f_4+f_5 \\
5 & f_7 & \newf_7 = f_7+f_8 \\
  \end{array}
\]
The new equations in the resulting DAE are
\begin{equation*}
\begin{aligned}
0=\newf_1&= R_0^{-1}(x_1-U_e(t)) - R_2^{-1}U_b + x_2\left( R_1^{-1} + R_2^{-1}\right) -(\alpha-1)g(x_2-x_3) \\
0=\newf_4 &= R_4^{-1}(x_4-U_b) + \alpha g(x_2-x_3)- R_5^{-1}U_b + x_5\left(R_5^{-1}+R_6^{-1}\right)  \\
&\phantom{=R_4^{-1}(x_4\;}
-(\alpha-1)g(x_5-x_6)\\
0=\newf_7&= R_8^{-1}(x_7-U_b) + \alpha g(x_5-x_6) + R_9^{-1}x_8\;.
\end{aligned}
\end{equation*}
{\scriptsize
\begin{align*}
\begin{blockarray}{r@{\hspace{1mm}} *8{@{\hspace{0.5mm}} c @{\hspace{0.5mm}}} @{\hspace{0.5mm}}cc}
 & x_{7} & x_{8} & x_{4} & x_{5} & x_{6} & x_{1} & x_{2} & x_{3} & \s{c_i}  \\
\begin{block}{r @{\hspace{2mm}}[*8{@{\hspace{0.5mm}} c@{\hspace{0.5mm}}} @{\hspace{0.5mm}}] cc}
 \newf_{7}&0^\bullet&0\phanbull&\dbar&0&\phanbull0&&&&\s{1}\\
 f_{8}&1\phanbull&1^\bullet&\dbar&&&&&&\s{0}\\\cline{2-5}
 \newf_{4}&&&\lbar{0^\bullet}&0&\lbar{0}&&0&0 &\s{1}\\
 f_{5}&&&\lbar{1\phanbull}&\phanbull1^\bullet&\lbar{\OK{0}}&&&&\s{0}\\\cline{4-6}
 f_{6}&&&&\OK{0}&\lbar{\phanbull1^\bullet}&\dbar&&&\s{0}\\\cline{6-8}
 \newf_{1}&&&&&&\lbar{\phanbull0^\bullet}&0&\lbar{0\phanbull }&\s{1}\\
 f_{2}&&&&&&\lbar{1}&\phanbull1^\bullet&\lbar{\OK{0}\phanbull }&\s{0}\\\cline{7-9}
 f_{3}&&&&&&&\OK{0}&\lbar{\ \ 1^\bullet \ \ }&\s{0}\\
\end{block}
\s{d_j} &\s{1}&\s{1}&\s{1}&\s{1}&\s{1}&\s{1}&\s{1}&\s{1}
\end{blockarray} 
\begin{blockarray}{r *8{@{\hspace{1mm}} c @{\hspace{1mm}}} c}
 &x'_{7} &x'_{8} &x'_{4} &x'_{5} &x'_{6} &x'_{1} &x'_{2} &x'_{3} \\[0.5ex]
\begin{block}{r @{\hspace{1mm}}[*8{@{\hspace{1mm}} c@{\hspace{1mm}}} ]c}
 \newf'_{7}&R_8^{-1}&R_9^{-1}&\lbar{}& \pp{\newf_7}{x_5} & \pp{\newf_7}{x_6} &&&\\
 f_{8}&-C_5&C_5&\lbar{}&&&&&\\ \cline{2-5} 
 &&&\lbar{}&&\lbar{}\\[-2ex]
 \newf'_{4}&&&\lbar{\ \ R_4^{-1}\ \ }&R_5^{-1}+R_6^{-1}&\dbar{$\pp{\newf_4}{x_6}$}&&\pp{\newf_4}{x_2}&\pp{\newf_4}{x_3}\\
 f_{5}&&&\lbar{-C_3}&C_3&\dbar&&&\\\cline{4-6}
 f_{6}&&&&&\lbar{C_4}&\dbar&&\\\cline{6-8}
  &&&&&&\lbar{}&&\lbar{}\\[-2ex]
 \newf'_{1}&&&&&&\lbar{\ \ R_0^{-1}\ \ }&R_1^{-1}+R_2^{-1}&\dbar{$\pp{\newf_1}{x_3}$}\\
 f_{2}&&&&&&\lbar{-C_1}&C_1&\dbar\\\cline{7-9}
 f_{3}&&&&&&&&\lbar{\ C_2\ \ }\\
\end{block}
\end{blockarray}
\end{align*}
}
The SA still reports index 1, and succeeds with a nonzero constant $\det(\newJ)$:
\[
\det(\newJ) = C_1C_2C_3C_4C_5
\left(R_0^{-1}+R_1^{-1}+R_2^{-1}\right)
\left(R_4^{-1}+R_5^{-1}+R_6^{-1}\right)
\left(R_8^{-1}+R_9^{-1}\right)  \neq 0\;.
\]
Now $\val\newSig=5<8=\val\Sig$.

\medskip
\noindent{\em ES method.} We can take $\~\hv=[1,1]^T\in \ker{\Jblk{1}{1}}, \ker{\Jblk{3}{3}}, \ker{\Jblk{5}{5}}$. We show how to perform a conversion on block 1; block 3 and block 5 can be treated in the same way.

For block 1, we construct the corresponding $\vecv=[1,1,\~0_8^T]^T$. Using \rf{ESblkdef}, we have
\begin{align*}
\nzsetES = \newJes = \setbg{ j \mid v_j \nequiv 0} = \set{1,2},\quad 
s=|\nzsetES|=2,\quad
\eqsetIES = \set{1,2},
\qaq
\ESconst &= 0\;.
\end{align*}
We choose $\indxk=1\in\newJes$, introduce for $x_2$ a new variable
\[
y_2 = x_2^{(d_2-\ESconst)} - \frac{v_2}{v_1} \cdot x_1^{(d_1-\ESconst)}  = x'_2-x'_1\;,
\]
and append correspondingly the equation $0=h_2=-y_2+x'_2-x'_1$. Then we replace $x'_2$ by $y_2+x'_1$ in $f_1, f_2$.

After we complete similar conversions on block 3 and block 5, the resulting DAE has equations $f_3, f_6$ and the following equations:
\begin{equation*}
\begin{alignedat}{7}
0&=\newf_1&&=-&&C_1 y_2 + R_0^{-1}(x_1-U_e(t)) \\
0&=\newf_2&&=&&C_1 y_2 - R_2^{-1}U_b + x_2\left(R_1^{-1}+R_2^{-1}\right) -(\alpha-1)g(x_2-x_3)\\
0&= h_2 &&=-&&y_2+x'_2-x'_1 \\ 
0&=\newf_4&&= -&&C_3y_5 + R_4^{-1}(x_4-U_b) + \alpha g(x_2-x_3)\\
0&=\newf_5&&= &&C_3y_5 - R_5^{-1}U_b + x_5\left(R_5^{-1}+R_6^{-1}\right) -(\alpha-1)g(x_5-x_6)\\
0&= h_5 &&=-&&y_5+x'_5-x'_4 \\ 
0&=\newf_7&&= -&&C_5 y_8 + R_8^{-1}(x_7-U_b) + \alpha g(x_5-x_6)\\
0&=\newf_8&&=  &&C_5 y_8 + R_9^{-1}x_8 \\
0&= h_8 &&=-&&y_8+x'_8-x'_7\;.
\end{alignedat}
\end{equation*}
{\scriptsize
\begin{align*}
\begin{alignedat}{3}
\begin{blockarray}{r@{\hspace{1mm}}  *{11}{@{\hspace{0.5mm}} c@{\hspace{0.5mm}} } @{\hspace{0.5mm}}  cc}
 & x_{7} & x_{8} & y_{8} & x_{4} & x_{5} & y_{5} & x_{6} & x_{1} & x_{2} & y_{2} & x_{3} & \s{c_i} \\
\begin{block}{r@{\hspace{1mm}}   [*{11}{@{\hspace{0.5mm}} c@{\hspace{0.5mm}} }] @{\hspace{0.5mm}} cc}
 h_{8}&\phanbull1^\bullet&1&\OK{0}&\dbar&&&&&&&&\s{0}\\
 \newf_{8}&&\phanbull0^\bullet&0&\dbar&&&&&&&&\s{1}\\
 \newf_{7}&0&&\phanbull0^\bullet&\dbar&0&&0&&&&&\s{1}\\\cline{2-7}
 h_{5}&&&&\lbar{\phanbull1^\bullet}&1&\OK{0}&\dbar&&&&&\s{0}\\
 \newf_{5}&&&&\dbar&\phanbull0^\bullet&0&\lbar{0}&&&&&\s{1}\\
 \newf_{4}&&&&\lbar{0}&&\phanbull0^\bullet&\dbar&&0&&0&\s{1}\\\cline{5-8}
 f_{6}&&&&&\OK{0}&&\lbar{\phanbull1^\bullet}&\dbar&&&&\s{0}\\\cline{8-11}
 h_{2}&&&&&&&&\lbar{\phanbull1^\bullet}&1&\OK{0}&\dbar&\s{0}\\
 \newf_{2}&&&&&&&&\dbar&\phanbull0^\bullet&0&\lbar{0}&\s{1}\\
 \newf_{1}&&&&&&&&\lbar{0}&&\phanbull0^\bullet&\dbar&\s{1}\\\cline{9-12}
 f_{3}&&&&&&&&&\OK{0}&&\lbar{\ \phanbull1^\bullet\ \ }&\s{0}\\
\end{block}
\s{d_j} &\s{1}&\s{1}&\s{1}&\s{1}&\s{1}&\s{1}&\s{1}&\s{1}&\s{1}&\s{1}&\s{1} \\
\end{blockarray}
\end{alignedat}
\begin{alignedat}{3}
\begin{blockarray}{r@{\hspace{1mm}}  *{11}{@{\hspace{0.5mm}} c@{\hspace{0.5mm}} } @{\hspace{0.5mm}}  cc}
 &x'_{7} &x'_{8} &y'_{8} &x'_{4} &x'_{5} &y'_{5} &x'_{6} &x'_{1} &x'_{2} &y'_{2} &x'_{3} \\
 \begin{block}{r@{\hspace{1mm}}   [*{11}{@{\hspace{0.5mm}} c@{\hspace{0.5mm}} }] @{\hspace{0.5mm}} cc}
 h_{8}&-1&1&&\dbar&&&&&&&\\
 \newf'_{8}&&R_9^{-1}&\phantom{-}C_5&\dbar&&&&&&&\\
 \newf'_{7}&R_8^{-1}&&-C_5&\dbar&\pp{\newf_7}{x_5}&&\pp{\newf_7}{x_6}&&&&\\\cline{2-7}
 h_{5}&&&&\lbar{-1}&1&&\dbar&&&&\\
 \newf'_{5}&&&&\dbar&\pp{\newf_5}{x_5}&\phantom{-}C_3&\lbar{\pp{\newf_5}{x_6}}&&&&\\
 \newf'_{4}&&&&\lbar{R_4^{-1}}&&-C_3&\dbar&&\pp{\newf_4}{x_2}&&\pp{\newf_4}{x_3}\\\cline{5-8}
 f_{6}&&&&&&&\lbar{C_4}&\dbar&&&\\\cline{8-11}
 h_{2}&&&&&&&&\lbar{-1}&1&&\dbar\\
 \newf'_{2}&&&&&&&&\dbar&\pp{\newf_2}{x_2}&\phantom{-}C_1&\lbar{\ \ \pp{\newf_2}{x_3}\ \ }\\
 \newf'_{1}&&&&&&&&\lbar{R_0^{-1}}&&-C_1&\dbar\\\cline{9-12}
 f_{3}&&&&&&&&&&&\lbar{C_3}\\
\end{block}
\end{blockarray}
\end{alignedat}
\end{align*}
}
Here in $\newJac$, $\ppin{\newf_5}{x_5}=R_5^{-1}+R_6^{-1}$ and $\ppin{\newf_2}{x_2}=R_1^{-1}+R_2^{-1}$. The SA succeeds with a nonzero constant $\det(\newJac)$ and $\val\newSig=5<8=\val\Sig$.

\subsection{Ring modulator.}\label{ssc:ringmod}

We study the ring modulator problem from \cite{TestSetIVP}. When $C_s\neq 0$, it is a stiff ODE system  of  15 nonlinear equations.
Setting $C_s=0$ gives a DAE of differentiation index 2 that consists of 11 differential and 4 algebraic equations:
\begin{alignat*}{7}
0&= f_1 &&= -&&x'_1 + C^{-1}\bigl( x_8-0.5x_{10}+0.5x_{11}+x_{14}-R^{-1}x_1 \bigr)\\
0&= f_2 &&= -&&x'_2 + C^{-1}\bigl( x_9-0.5x_{11}+0.5x_{12}+x_{15}-R^{-1}x_2\bigr)\\
0&= f_3 &&= &&x_{10}-q(U_{D1})+q(U_{D4})\\
0&= f_4 &&= -&&x_{11}+q(U_{D2})-q(U_{D3}) \\
0&= f_5 &&= &&x_{12}+q(U_{D1})-q(U_{D3})\\
0&= f_6 &&= -&&x_{13}-q(U_{D2})+q(U_{D4})\\
0&= f_7 &&= -&&x'_7 + C_p^{-1}\bigl(-R_p^{-1}x_7+q(U_{D1})+q(U_{D2})-q(U_{D3})-q(U_{D4}) \bigr)
\\0&= f_8 &&= -&&x'_8 + -L_h^{-1}x_1\\
0&= f_9 &&= -&&x'_9 + -L_h^{-1}x_2\\
0&= f_{10} &&= -&&x'_{10} + L_{s2}^{-1}(0.5x_1-x_3-R_{g2}x_{10}) \\
0&= f_{11} &&= -&&x'_{11} + L_{s3}^{-1}(-0.5x_1+x_4-R_{g3}x_{11}) \\
0&= f_{12} &&= -&&x'_{12} + L_{s2}^{-1}(0.5x_2-x_5-R_{g2}x_{12})\\ 
0&= f_{13} &&= -&&x'_{13} + L_{s3}^{-1}(-0.5x_2+x_6-R_{g3}x_{13})\\
0&= f_{14} &&= -&&x'_{14} + L_{s1}^{-1}(-x_1+U_{in1}(t)-(R_i+R_{g1})x_{14})\\
0&= f_{15} &&= -&&x'_{15} + L_{s1}^{-1}(-x_2-(R_c+R_{g1})x_{15})\;.
\end{alignat*}
The functions are
\[
\begin{alignedat}{5}
U_{D1} &= &&x_3-x_5-x_7-U_{in2}(t)
&\qquad\qquad&  &q(U)&= \gamma (e^{\delta U}-1)
\\
U_{D2} &= -&&x_4+x_6-x_7-U_{in2}(t)
&& &U_{in1}(t)&= 0.5 \sin (2000\pi t)
\\
U_{D3} &= &&x_4+x_5+x_7+U_{in2}(t)
&&  &U_{in2}(t)&= 2 \sin (20000\pi t)
\\
U_{D4} &= -&&x_3-x_6+x_7+U_{in2}(t)\;.
\end{alignedat}
\]
We refer the reader to \cite{TestSetIVP} for the nonzero constants $C$, $C_p$, $R$, $R_p$, $R_c$, $\gamma$, $L_h$, $L_{s1}$, $L_{s2}$, $L_{s3}$, $R_{g1}$, $R_{g2}$, $R_{g3}$, $R_{i}$, and $\delta$.
\begin{equation*}
{
\Sig =
\begin{blockarray}{r @{\hspace{2mm}} *{15}{@{\hspace{1mm}} c @{\hspace{1mm}}}  @{\hspace{1mm}}cc}
 & x_{1} & x_{2} & x_{7} & x_{13} & x_{11} & x_{12} & x_{10} & x_{3} & x_{4} & x_{5} & x_{6} & x_{8} & x_{9} & x_{14} & x_{15} & \s{c_i}  \\
\begin{block}{r @{\hspace{2mm}}  [*{15}{@{\hspace{1mm}} c @{\hspace{1mm}}}] @{\hspace{1mm}}cc}
 f_{1}&1^\bullet&\dbar&&&\OK{0}&&\OK{0}&&&&&\OK{0}&&\OK{0}&&\s{0}\\\cline{2-3}
 f_{2}&&\lbar{1^\bullet}&\dbar&\OK{0}&&\OK{0}&&&&&&&\OK{0}&&\OK{0}&\s{0}\\\cline{3-4}
 f_{7}&&&\lbar{1^\bullet}&\dbar&&&&0&0&0&0&&&&&\s{0}&\\\cline{4-5}
 f_{13}&&\OK{0}&&\lbar{1^\bullet}&\dbar&&&&&&0&&&&&\s{0}\\\cline{5-6}
 f_{11}&\OK{0}&&&&\lbar{1^\bullet}&\dbar&&&0&&&&&&&\s{0}\\\cline{6-7}
 f_{12}&&\OK{0}&&&&\lbar{1^\bullet}&\dbar&&&0&&&&&&\s{0}\\\cline{7-8}
 f_{10}&\OK{0}&&&&&&\lbar{1^\bullet}&\lbar{0}&&&&&&&&\s{0}\\\cline{8-12}
 f_{3}&&&\OK{0}&&&&\OK{0}&\lbar{0^\bullet}&&0&0&\dbar&&&&\s{0}\\
 f_{4}&&&\OK{0}&&\OK{0}&&&\dbar&0^\bullet&0&0&\dbar&&&&\s{0}\\
 f_{5}&&&\OK{0}&&&\OK{0}&&\lbar{0}&0&0^\bullet&&\dbar&&&&\s{0}\\
 f_{6}&&&\OK{0}&\OK{0}&&&&\lbar{0}&0&&0^\bullet&\dbar&&&&\s{0}\\\cline{9-13}
 f_{8}&\OK{0}&&&&&&&&&&&\lbar{1^\bullet}&\dbar&&&\s{0}\\\cline{13-14}
 f_{9}&&\OK{0}&&&&&&&&&&&\lbar{1^\bullet}&\dbar&&\s{0}\\\cline{14-15}
 f_{14}&\OK{0}&&&&&&&&&&&&&\lbar{1^\bullet}&\dbar&\s{0}\\\cline{15-16}
 f_{15}&&\OK{0}&&&&&&&&&&&&&\lbar{1^\bullet}&\s{0}\\
\end{block}
\s{d_j} &\s{1}&\s{1}&\s{1}&\s{1}&\s{1}&\s{1}&\s{1}&\s{0}&\s{0}&\s{0}&\s{0}&\s{1}&\s{1}&\s{1}&\s{1}
\end{blockarray}
}
\end{equation*}
Each block of size 1 has a nonsingular sub-Jacobian $-1$. 
Block 8 has an identically singular sub-Jacobian
\[
\Jblk{8}{8} = 
\begin{blockarray}{r@{\hspace{2mm}} *4{@{\hspace{1mm}} c@{\hspace{1mm}}}}
 & x_{3} &x_{4} &x_{5} &x_{6}  \\
\begin{block}{r @{\hspace{2mm}}[*4{@{\hspace{1mm}} c@{\hspace{1mm}}}]}
f_3 &-s_1-s_4 &  & s_1 & -s_4 \\
f_4 & &-s_2-s_3 & -s_3 & s_2\\
f_5 & s_1 & -s_3&-s_1-s_3&\\
f_6 & -s_4 & s_2 & & -s_2-s_4\\
\end{block}
\end{blockarray}\;, \quad\text{where}\quad s_i = \gamma\delta e^{\delta U_{Di}}\;.
\]
 This is a nonlinear block, since variables $x_3, x_4, x_5, x_6$ do not occur jointly linearly in equations $f_3, f_4, f_5, f_6$. One can also see these variables appear in $\Jblk{8}{8}$.

\medskip
\noindent{\em LC method.}
We find a constant vector $\~\hu = [1,-1,1,-1]^T\in\coker{\Jblk{8}{8}}$, which satisfies the block LC condition \rf{LCblkcond}. Then $\vecu = [\~0_7^T, 1,-1,1,-1, \~0_4^T]^T$. We use \rf{LCblkdef} to derive
\[
\eqsetI = \setbg{ i\mid\, u_i\nequiv 0\,} = \set{8,9,10,11},
\quad \LCconst=0, \qaq
\nzset = \newL =  \set{8,9,10,11} \;.
\]
The row indices in $\newL$ correspond to the equations $f_3, f_4, f_5, f_6$. We can pick any one of them and replace it by
\begin{align*}
\newf &= u_1f_3+u_2f_4+u_3f_5+u_4f_6 = f_3-f_4+f_5-f_6
= x_{10} + x_{11} + x_{12} + x_{13}\;.
\end{align*}
We choose $f_3$ and replace it by $\newf_3=\newf$. The resulting DAE has the following $\newSig$ with $\val\newSig=10<11=\val\Sig$.
\begin{equation*}
{
\newSig =
\begin{blockarray}{r @{\hspace{2mm}}*{15}{@{\hspace{1mm}}c@{\hspace{1mm}}} @{\hspace{1mm}}c}
 & x_{1} & x_{2} & x_{7} & x_{3} & x_{4} & x_{5} & x_{6} & x_{10} & x_{11} & x_{12} & x_{13} & x_{8} & x_{9} & x_{14} & x_{15} & \s{c_i}  \\
\begin{block}{r @{\hspace{2mm}}[*{15}{@{\hspace{1mm}}c@{\hspace{1mm}}}] @{\hspace{1mm}}c}
 f_{1}&1^\bullet&\dbar&&&&&&\OK{0}&\OK{0}&&&\OK{0}&&\OK{0}&&\s{0}\\\cline{2-3}
 f_{2}&&\lbar{1^\bullet}&\dbar&&&&&&&\OK{0}&\OK{0}&&\OK{0}&&\OK{0}&\s{0}\\\cline{3-4}
 f_{7}&&&\lbar{1^\bullet}&\lbar{0}&0&0&0&&&&&&&&&\s{0}\\\cline{4-12}
 f_{10}&\OK{0}&&&\lbar{0^\bullet}&&&&1&&&&\dbar&&&&\s{0}\\
 f_{5}&&&\OK{0}&\lbar{0}&0^\bullet&0&&&&\OK{0}&&\dbar&&&&\s{0}\\
 f_{4}&&&\OK{0}&\dbar&0&0^\bullet&0&&\OK{0}&&&\dbar&&&&\s{0}\\
 f_{6}&&&\OK{0}&\lbar{0}&0&&0^\bullet&&&&\OK{0}&\dbar&&&&\s{0}\\
  \newf_{3}&&&&\dbar&&&&0^\bullet&0&0&0&\dbar&&&&\s{1}\\
 f_{11}&\OK{0}&&&\dbar&0&&&&1^\bullet&&&\dbar&&&&\s{0}\\
 f_{12}&&\OK{0}&&\dbar&&0&&&&1^\bullet&&\dbar&&&&\s{0}\\
 f_{13}&&\OK{0}&&\dbar&&&0&&&&1^\bullet&\dbar&&&&\s{0}\\\cline{5-13}
 f_{8}&\OK{0}&&&&&&&&&&&\lbar{1^\bullet}&\dbar&&&\s{0}\\\cline{13-14}
 f_{9}&&\OK{0}&&&&&&&&&&&\lbar{1^\bullet}&\dbar&&\s{0}\\\cline{14-15}
 f_{14}&\OK{0}&&&&&&&&&&&&&\lbar{1^\bullet}&\dbar&\s{0}\\\cline{15-16}
 f_{15}&&\OK{0}&&&&&&&&&&&&&\lbar{1^\bullet}&\s{0}\\
\end{block}
\s{d_j} &\s{1}&\s{1}&\s{1}&\s{0}&\s{0}&\s{0}&\s{0}&\s{1}&\s{1}&\s{1}&\s{1}&\s{1}&\s{1}&\s{1}&\s{1} \\
\end{blockarray}
}
\end{equation*}

Again, each $1\times 1$ block has a nonsingular Jacobian:
$\ppin{f_i}{x'_i}=-1$ for $i=$1, 2, 7, 8, 9, 14, 15.
The sub-Jacobian of block 4 in the resulting DAE is
\begin{align*}
\newJblk{4}{4} = 
\begin{blockarray}{r@{\hspace{2mm}} *8{@{\hspace{1mm}} c @{\hspace{1mm}}}}
 &x_{3} &x_{4} &x_{5} &x_{6} &x'_{10} &x'_{11} &x'_{12} &x'_{13}  \\
\begin{block}{r @{\hspace{2mm}}[*8{@{\hspace{1mm}} c @{\hspace{1mm}}}]}
 f_{10}&-L_{s2}^{-1}&&&&-1&&&\\
 f_{5}&s_1 & -s_3&-s_1-s_3&&&&&\\
 f_{4}& &-s_2-s_3 & -s_3 & s_2&&&&\\
 f_{6}&-s_4 & s_2 & & -s_2-s_4&&&&\\
 \newf'_{3}&&&&&1&1&1&1\\
 f_{11}&&L_{s3}^{-1}&&&&-1&&\\
 f_{12}&&&-L_{s2}^{-1}&&&&-1&\\
 f_{13}&&&&L_{s3}^{-1}&&&&-1 \\
\end{block}
\end{blockarray}\;,
\end{align*}
whose determinant is
$
\det(\newJ_{44})=2s_1s_2s_3s_4
(s_1^{-1} + s_2^{-1} + s_3^{-1} + s_4^{-1})
(L_{s2}^{-1} + L_{s3}^{-1})
$.
The SA succeeds at any point where $\det(\newJ_{44})\neq0$, and the DAE is of index 2.


\medskip
\noindent{\em ES method.}
Find $\~\hv=[-1,1,-1,1]^T\in \ker{\Jblk{8}{8}}$. Then $\vecv=[\~0_7^T, -1, 1, -1, 1, \~0_4^T]^T$.
We use \rf{ESblkdef} to derive
\begin{align*}
\nzsetES = \newJes = \setbg{ j \mid v_j\nequiv 0 } = \set{8,9,10,11}, \quad 
s = |\nzsetES| = 4,
\quad \eqsetIES = \nzsetES,
\qaq \ESconst  = 0
\;.
\end{align*}

We choose column index $\indxk=8\in\newJes$ in the permuted $\Sig$. The variable of this column is $x_3$. The other variables in block 8 are $x_4$, $x_5$, $x_6$, so we introduce for them, respectively,
\[
y_4 = x_4 - (\fracin{v_9}{v_8})\cdot x_3, \quad 
y_5 = x_5 - (\fracin{v_{10}}{v_8})\cdot x_3, \qaq
y_6 = x_6 - (\fracin{v_{11}}{v_8})\cdot x_3\;.
\]
Then we append the equations corresponding to these variables 
\[
0 = g_4 = - y_4 + x_4 + x_3, \quad
0 = g_5 = - y_5 + x_5 - x_3, \qaq
0 = g_6 = - y_6 + x_6 + x_3\;.
\]
The equations in block 8 are $f_3$, $f_4$, $f_5$, $f_6$. In these equations, we perform the following substitutions.
\[
  \begin{array}{l@{\hspace{5mm}}l@{\hspace{5mm}}l}
 \text{replace} & \text{by} &\text{in}  \\ \hline
x_4 & y_4-x_3 & f_4, f_5, f_6 \\
x_5 & y_5+x_3 & f_3, f_4, f_5 \\
x_6 & y_6-x_3 & f_3, f_4, f_6 
  \end{array}
\]
The resulting index-2 DAE is of size 18. (We do not display the results of SA here.) It has $\val\newSig = 10 <11=\val\Sig$ and $\det(\newJ)=-2s_1s_2s_3s_4(s_1^{-1} + s_2^{-1} + s_3^{-1} + s_4^{-1})(L_{s2}^{-1} + L_{s3}^{-1})$. The largest fine block is of size 12, and the other six fine blocks are of size 1. The SA succeeds at any point where $\det(\newJ)\neq0$.

\section{Conclusions.}\label{sc:conclu}

We combined block triangularization with the LC and ES conversion methods for improving the \Sigmeth. When $\Jac$ is identically singular and the DAE has a nontrivial BTF, we can locate each diagonal block whose corresponding sub-Jacobian is identically singular, and perform a conversion on it. We base this strategy on the view that each diagonal block can be regarded as a sub-DAE, while formulas contributing to the off diagonal blocks are regarded as driving terms. 

Compared with the basic conversion methods that work on the whole DAE, the block methods only work on singular blocks, which are usually smaller than the DAE itself. Hence the block methods require fewer symbolic computations, and can generally find a useful conversion for reducing $\val\Sig$ more efficiently. As in the basic case, a conversion applied on a singular block guarantees (a) a strict decrease in the value of the (whole) signature matrix, and (b) the equivalence between the original DAE and the resulting one.  The rationale for choosing a desirable conversion method is in \cite[Table 4.1]{tgn2015c}. 

We combine \matlab's Symbolic Math Toolbox \cite{matlabsymbolictb} with our structural analysis software \daesa \cite{NedialkovPryce2012b,NedialkovPryce2012a}, and have built a prototype code that automates the conversion process. We aim to incorporate them in a future version of \daesa.

With our prototype code, we have applied our methods on numerous DAEs on which the \Sigmeth fails. They are either arbitrarily constructed to be SA-failure cases for our investigations, or borrowed from the existing literature. Our conversion methods succeed in fixing all these solvable DAEs.
We believe that our assumptions and conditions are reasonable for practical problems, and that these methods can help make the \Sigmeth more reliable. 

We end these two articles with our main conjecture regarding SA's failure. In all our experiments, when we successfully fix the failure using our conversion methods, the value of a signature matrix always decreases. As Pryce points out in \cite{Pryce98}, the solvability of a DAE lies within its inherent nature, not the way it is formulated or analyzed. Hence we conjecture that a DAE formulation friendly to SA should have a reasonable but never overestimated $\val\Sig$ that can be interpreted as the number of degrees of freedom (DOF) of the underlying problem. In other words, a DAE should not be formulated to exhibit more DOF than the underlying problem has. However, based on our current knowledge, it appears difficult to show why overestimating DOF can lead to an identically singular System Jacobian.

\appendix
\section{Proof of \LEref{ESblklemma}.}\label{sc:proofESblknewSig}

For $\newSig=(\newsij{i}{j})$ in the block structure in \FGref{ESblkSig}, we write the block sizes in the array 
\begin{align*}
\wN{} = (N_1,N_2,\ldots,N_{\idq-1},N_\idq+s,N_{\idq+1},\ldots,N_p)\;,
\end{align*}
and also write the block sizes of $\Sig$ in the array $N = (N_1,N_2,\ldots,N_{\idq-1},N_\idq,N_{\idq+1},\ldots,N_p)$.
Let $\wblk{i}$ denote the block number of a row or column index $i$ in $\newSig$. From the construction of $\wN{}$ and $N$, it is not difficult to show that
\begin{align}
\wblk{j}<\idq 
&\;\Leftrightarrow\; 1\le j\le \sum_{w=1}^{q-1} N_w 
\;\Leftrightarrow\; \blk{j}<\idq \quad\text{and}
\label{eq:blockOfj<q}
\\
\wblk{j+s}>\idq
&\;\Leftrightarrow\; \sum_{w=1}^{q} N_w +s+1\le j+s 
\;\Leftrightarrow\; \blk{j}>\idq\;.
\label{eq:blockOfj>q}
\end{align}

From the construction of $(\~\newc;\~\newd)$ in \rf{ESblknewcd}, each variable $x_j$ for $j=\rnge{1}{n}$ has the same ``variable offset'' in $\newSig$ as $x_j$ has in $\Sig$. Also, each equation $\newf_i$ for $i=\rnge{1}{n}$ has the same ``equation offset'' in $\newSig$ as $f_i$ has in $\Sig$. Quotation marks are used here because $(\~\newc;\~\newd)$ is {\em not} a valid offset pair of $\newSig$; this vector pair is merely used for proving $\val\newSig<\val\Sig$ in \THref{ESblk}.

We aim to show that
\begin{align}\label{eq:ESblkcd}
\newd_j - \newc_i\ 
\left\{
\begin{alignedat}{7}
&>\newsij{i}{j}     &\quad& \text{if $\wblk{j} < \wblk{i}$}\\
&\ge \newsij{i}{j}  && \text{if $\wblk{j} \ge \wblk{i}$}\;.
\end{alignedat}
\right.
\end{align}

For the block structure of $\newSig$ in \FGref{ESblkSig}, we have shown on page~\pageref{fg:ESblkSig} that $\newSig_{\idv_1 \idv_2}=\Sig_{\idv_1 \idv_2}$ if $\idv_1\neq \idq$ and $\idv_2\neq\idq$, and that $\newSig_{\idv_1 \idv_2}=\bigl[ \Sig_{\idv_1\idq},\ \bninfty_{N_{\idv_1}\times s} \bigr] $ if $\idv_1\neq \idq$ and $\idv_2=\idq$.
Hence, provided $\idv_1\neq \idq$, $\newSig_{\idv_1 \idv_2}$ is below [resp. above] the block diagonal of $\newSig$, if $\Sig_{\idv_1 \idv_2}$ is below [resp. above] the block diagonal of $\Sig$.
By \rf{djciblk}, the inequalities in \rf{ESblkcd} hold for $i$ with $\wblk{i}\neq\idq$.

What remains to show is the inequalities in \rf{ESblkcd} for $i$ with $\wblk{i}=\idq$. These inequalities are for the signature entries in  $\newSig_{\idq\idv_2}$, the blocks affected by the expression substitutions.  
We consider three cases for $\newSig_{\idq\idv_2}$: it is (a) below the block diagonal, with $\idv_2<\idq$, (b) above the block diagonal, with $\idv_2>\idq$, or (c) the diagonal block $\newSig_{\idq\idq}$, with $\idv_2=\idq$.

(a)\ \ $\newSig_{\idq \idv_2}$ with $\idv_2<\idq$. 
An entry $(i,j)$ in this block satisfies $\wblk{j}<\wblk{i}=\idq$. By \rf{blockOfj<q}, $\blk{j}<\idq$ and hence $j\notin B_\idq$. 

Recall from \rf{ESblksubs} that, in each $f_i$ with $i\in\eqsetIES\subseteq B_\idq$, we substitute $\Big( y_r + \frac{v_r}{v_\indxk} \cdot x_\indxk^{(d_\indxk-\ESconst)} \Big)^{(\ESconst-c_i)}$ for each $x_r^{(\sij{i}{r})}$ with $d_r-c_i=\sij{i}{r}$ and  $r\in\Jnol\subset B_\idq$.
For a $j\notin B_\idq\supset\Jnol$,  the corresponding derivatives $x_j^{(d_j-c_i)}$ are not replaced in the ES conversion, and for $r\in\Jnol$ (so $j\,,r\,,\indxk$ are distinct),
\begin{align}\label{eq:xjvcci}
\hoder{x_j}{\Big( y_\jtwo+\frac{v_\jtwo}{v_\indxk}\cdot x_\indxk^{(d_\indxk-\ESconst)} \Big)^{(\ESconst-c_i)}}
= \hoder{x_j}{\Big(\frac{v_\jtwo}{v_\indxk}\Big)^{(\ESconst-c_i)}}
\le  \hoder{x_j}{\vecv^{(\ESconst-c_i)}}
= \hoder{x_j}{\vecv}+(\ESconst-c_i)\;.
\end{align}
By \rf{ESblkcond}, $\hoder{x_j}{\vecv} < d_j - \ESconst$. Using \rf{djciblk} and \rf{xjvcci}, we derive
\begin{align}
\hoder{x_j}{\newf_i} &\le \max\left\{   \hoder{x_j}{f_i}, 
\max_{\jtwo\in\nzsetES\setminus\{\indxk\}}\hoder{x_j}{\Big( y_\jtwo+\frac{v_\jtwo}{v_\indxk}\cdot x_\indxk^{(d_\indxk-\ESconst)} \Big)^{(\ESconst-c_i)}}\right\} 
 \nonumber 
 \\&\le\max\left\{   \sij{i}{j},\  \hoder{x_j}{\vecv}+(\ESconst-c_i)  \right\} 
 \nonumber \\&
 <\max\left\{   d_j-c_i,\  (d_j-\ESconst)+(\ESconst-c_i) \right\} 
 = d_j-c_i
\quad\text{ for $i\in \eqsetIES\subseteq B_\idq$}\;.
\label{eq:less1a} 
\end{align}
From the ES conversion described in \THref{ESblk}, we have

\begin{align}
\hoder{x_j}{\newf_i}&=\hoder{x_j}{f_i}< d_j-c_i\quad\,\,\text{ for $i\in B_\idq\setminus\eqsetIES$ and}
 \label{eq:less2a} 
 \\\hoder{x_j}{g_r} &\le \hoder{x_j}{v}<d_j-\ESconst \quad\;\;\;\,\text{ for $r\in \nzsetES$}\;.
 \label{eq:less3a}
\end{align}
Since blocks $\newSig_{\idq \idv_2}$ with $\idv_2<\idq$ contain signature entries $\newsij{i}{j}$ for equations $\newf_i$ and $g_r$, where $i\in B_\idq$ and $r\in\nzsetES$, in variables $x_j$ with $\blk{j}<\idq$, by taking together the inequalities in \rf{less1a}-\rf{less3a},  we have  
\[
\newsij{i}{j}< \ 
\left\{
\begin{alignedat}{7}
&d_j-c_i     &\quad& \text{if $\blk{j} <  \idq$ and $i\in B_\idq$}\\
&d_j-\ESconst   && \text{if $\blk{j}<\idq$ and $i\in\rnge{\idQ+1}{\idQ+s}$}\;;
\end{alignedat}
\right.
\]
recall $\idQ = \sum_{w=1}^\idq N_w$.
Using \rf{blockOfj<q} and the construction of $(\~\newc;\~\newd)$ in \rf{ESblknewcd}, we have
\begin{align*}
\newsij{i}{j}<\newd_{j} - \newc_i \quad\text{for}\quad \wblk{j}<\wblk{i}=\idq\;.
\end{align*}
($\idq$ is the block number of both the original and enlarged diagonal blocks.)

\medskip

(b)\ \ 
$\newSig_{\idq \idv_2}$ with $\idv_2>\idq$. 
An entry $(i,j+s)$ in this block satisfies  $\wblk{j+s}>\wblk{i}=\idq$. 
By \rf{blockOfj>q}, $\blk{j}>\idq$ and hence $j\notin B_\idq\supset\Jnol$. By the same arguments as in (a), the corresponding derivatives $x_j^{(d_j-c_i)}$ are not replaced in the ES conversion.
 
By \rf{ESblkcond}, $\hoder{x_j}{\vecv} \le d_j - \ESconst$. Then by the same derivations as \rrf{less1a}{less3a} in (a), we have

\begin{align}
\hoder{x_j}{\newf_i} &\le d_j-c_i\qquad\qquad\qquad\;\,\text{ for $i\in \eqsetIES\subseteq B_\idq$\;,}  
\label{eq:less1b} 
\\\hoder{x_j}{\newf_i}&=\hoder{x_j}{f_i}\le d_j-c_i\quad\,\,\text{ for $i\in B_\idq\setminus\eqsetIES$, and}
 \label{eq:less2b} 
 \\\hoder{x_j}{g_r} &\le \hoder{x_j}{v}\le d_j-\ESconst \qquad\text{ for $r\in \nzsetES$}\;.
  \label{eq:less3b}
\end{align}
Since blocks $\newSig_{\idq \idv_2}$ with $\idv_2>\idq$ contain signature entries $\newsij{i}{,j+s}$ for equations $\newf_i$ and $g_r$, where $i\in B_\idq$ and $r\in\nzsetES$, in variables $x_j$ with $\blk{j}>\idq$, the inequalities \rrf{less1b}{less3b} yield
\[
\newsij{i}{,j+s}\le \ 
\left\{
\begin{alignedat}{7}
&d_j-c_i     &\quad& \text{if $\blk{j}>\idq$ and $i\in B_\idq$}\\
&d_j-\ESconst   && \text{if $\blk{j}>\idq$ and $i\in\rnge{\idQ+1}{\idQ+s}$}\;.
\end{alignedat}
\right.
\]
Using again \rf{blockOfj>q} and $(\~\newc;\~\newd)$ in \rf{ESblknewcd}, we have $\newsij{i}{,j+s}\le\newd_{j+s} - \newc_i$ for $\wblk{j+s}>\wblk{i}=\idq$, with $j=\rnge{\idQ+1}{n}$.
We can rewrite this inequality as
\begin{align*}
\newsij{i}{j}\le\newd_{j} - \newc_i \quad\text{for}\quad \wblk{j}>\wblk{i}=\idq\;,
\end{align*}
with $j=\rnge{\idQ+1+s}{n+s}$.

\medskip

(c)\ \ $\newSig_{\idq \idv_2}$ is $\newSig_{\idq\idq}$, with $\idv_2=\idq$.  
An entry $(i,j)$ in $\newSig_{\idq\idq}$ satisfies $\wblk{j}=\wblk{i}=\idq$. We view block $\idq$ of the original DAE as a sub-DAE, with a signature matrix $\Sig_{\idq\idq}$ of size $N_\idq$ and an offset pair $(\~c_\idq;\~d_\idq)$. Given that the ES conditions are satisfied by \rf{ESblkcond}, performing the ES conversion as described in \THref{ESblk} is equivalent to applying the basic ES method to this sub-DAE. 
After a conversion, the resulting enlarged signature matrix $\newSig_{\idq\idq}$ of size $N_\idq+s$ has the form
\begin{equation}\label{eq:ESblkqq}
\newSig_{\idq\idq}=
\left[
\begin{array}{c|c|c}
\newSig_{\idq\idq,11} &\newSig_{\idq\idq,12} &\newSig_{\idq\idq,13} \\[1ex]  \hline \\[-2ex]
\newSig_{\idq\idq,21} &\newSig_{\idq\idq,22} &\newSig_{\idq\idq,23} 
\end{array}
\right]\;;
\end{equation}
cf. \FGref{ESblkSig} and the details of the basic ES method in \cite{tgn2015c}.
The two block rows of $\newSig_{\idq\idq}$ correspond to 
$\newf_i$ for $i\in B_\idq$ and 
$g_j$ for $j\in\nzsetES$, respectively.
The three block columns of $\newSig_{\idq\idq}$ correspond to
$x_j$ for $j\in\nzsetES$,
$x_j$ for $j\in B_\idq\setminus\nzsetES$, and
$y_j$ for $j\in\nzsetES$, respectively. 
If we apply the same arguments in the proof of \defterm{Lemma 4.4} for the basic ES method, then we have $\newd_{j} - \newc_i \ge \newsij{i}{j}$ for all entries in $\newSig_{\idq\idq}$.  \qed

\begin{acknowledgements}
The authors acknowledge with thanks the financial support for this research: GT is supported in part by the McMaster Centre for Software Certification through the Ontario Research Fund, Canada, NSN is supported in part by the Natural Sciences and Engineering Research Council of Canada, and JDP is supported in part by the Leverhulme Trust, the UK. 
\end{acknowledgements}

\end{document}